\documentclass{IEEEtran}
\usepackage{graphicx}

\usepackage{cite}
\usepackage{amsmath}
\usepackage{amsthm}
\usepackage{mathtools}
\usepackage{subfigure}
\usepackage{float}
\usepackage{times}
\usepackage{enumerate}
\usepackage{amssymb}
\usepackage{algorithm}
\usepackage{algorithmic}
\usepackage{multirow}
\usepackage{multicol}
\usepackage{caption}
\usepackage{xcolor}

\newtheorem{lemma}{Lemma}
\newtheorem{definition}{Definition}
\newtheorem{proposition}{Proposition}
\newcommand{\overbar}[1]{\mkern 1.5mu\overline{\mkern-1.5mu#1\mkern-1.5mu}\mkern 1.5mu}

\begin{document}

\title{Stable Throughput Region of the Two-User Interference Channel}
\author{\IEEEauthorblockN{Nikolaos Pappas and Marios Kountouris}
		
	\thanks{N. Pappas is with the Department of Science	and Technology, Link{\"o}ping University, SE-60174 Norrk{\"o}ping, Sweden. Email: nikolaos.pappas@liu.se.}
	\thanks{M. Kountouris is with the Mathematical and Algorithmic Sciences Lab, Paris Research Center, Huawei Technologies Co. Ltd.
Boulogne-Billancourt, 92100, France. Email: marios.kountouris@huawei.com.}
}

\maketitle

\begin{abstract}
We consider the two-user interference channel where two independent pairs communicate concurrently and investigate its stable throughput region.
First, the stability region is characterized for the general case, i.e., without any specific consideration on the transmission and reception structures. Second, we explore two different interference harnessing strategies at the receiver: treating interference as noise and successive interference cancellation. Furthermore, we provide conditions for the convexity of the stability region and for which a certain receiver strategy leads to broader stability region. The impact of multiple transmit antennas on the stability region is briefly discussed. Finally, we study the effect of random access on the stability region of the two-user interference channel.
\end{abstract}

\section{Introduction} \label{sec:intro}
The theoretical foundation behind some of our society’s most advanced communication technologies can be found in information theory, which provided valuable insights and a roadmap to communication theorists and engineers. Information theory, as developed ever since Shannon's seminal paper, has mainly focused on single links (source-channel-destination) and centralized networks. A central question has been the maximization of the rate of reliable data transmission from a source to a destination; the concept of channel capacity, i.e., the boundary between the physically possible and physically impossible in terms of reliable data rate, has played a cardinal role. Nevertheless, the remarkable success of Shannon theory has not yet translated to communication networks. Extending the information-theoretic approach to a multi-terminal system involves the characterization of the maximum joint achievable rates at which different users sharing the channel can transmit (capacity region). Link-based information theory does not appear well-suited to the role of understanding network performance limits. This is partly because the capacity region is usually derived under the assumption of backlogged users and saturated (non-empty) queues. Considering the stochastic and bursty nature of traffic, other measures of rate performance become relevant and meaningful, such as the maximum stable throughput or stability region (in packets/slot). The stability region is defined as the union of all arrival rates for which all queue lengths stay finite or have a non-degenerate limiting probability distribution (there are several definitions of stability)~\cite{Szpankowski:stability}. These two regions are not in general identical and the conditions under which they coincide are known in very few cases \cite{AEUnion}.

There exists a vast literature related to the characterization of the stability region in wireless networks. Early work focuses on simple models, such as the collision channel and point-to-point communication link \cite{SastryNOW}. In this paper, we consider the two-user interference channel, which models communication scenarios with multiple point-to-point links transmitting over a common frequency band, thus generating interference to each other. The capacity region of the general Gaussian interference channel is a long standing problem and is only known for special cases, such as Gaussian channels with weak (``noisy") or strong interference~\cite{Carleial75, Shang09, AV09}. Furthermore, information-theoretic results advocate several ways of handling the interference, including orthogonal access, treating interference as noise (IAN), successive interference cancellation (SIC), joint decoding and interference alignment~\cite{b:Gamal_NIT}. There is a large amount of information- and communication-theoretic studies for the interference channels, including  multiuser cognitive interference channels \cite{MaamariTCCN2015}, power allocation \cite{ChaitanyaNCC2013}, and multi-antenna beamforming \cite{LindblomWCL2013,MISO_Pareto}. Nevertheless, previous studies usually assume infinite backlog and ignore the effect of bursty traffic. In this work, we investigate the stability region of the two-user interference channel, which, to the best of our knowledge, has not been reported to the literature.

\subsection{Related work}
In \cite{b:Simeone}, the authors studied a cognitive interference channel, as well as the case of a primary user and a cognitive user with and without relaying capabilities. The stable throughput region and the delay performance of a queue-aware multiple access scheme are studied in \cite{DimitriouTWC2018}. Using the effective capacity framework, \cite{QiaoTCOM2013} studies the throughput region of broadcast and interference channels under statistical delay constraints. The stability region exploiting past receptions is derived in \cite{PanACCESS2017} and \cite{JiaoIET2017} considers the problem of power allocation for the two-user interference channel with unsaturated traffic. The stability region of the two-user broadcast channel is derived in \cite{PappasTCOM2018} and the effect of multipacket reception on stability and delay of slotted ALOHA-based random access systems is considered in \cite{b:Naware}.
The stability region of the two-user interference channel is first reported in the conference version of this journal \cite{PappasITW2013}.

\subsection{Contributions}
In this work, we study the two-user interference channel (IC), where each user has bursty arrivals and transmits a packet whenever its queue is not empty. First, we obtain the exact stability region for the general case, i.e., without any specific consideration on the details of the transmission and interference handling schemes. 
The exact characterization of the stability region of networks with bursty sources is known to be a challenging problem due to the coupling among queues, i.e., the service process of a queue depends on the status of the other queues. To overcome this difficulty, we use the stochastic dominance technique~\cite{rao:stability}.
We also investigate two widely used interference handling schemes: (i) each receiver treats interference as noise, and (ii) successive interference cancellation is employed at the receiver side. Furthermore, we derive conditions for the convexity of the stability region and show under which system parameters, each interference management technique is superior (in the sense of broader stability region) compared to the other. Finally, we study the case of random access and we derive the stability conditions of the two-user interference channel under random access. The obtained results can be utilized for investigating the effect of successive interference cancellation in network level cooperative schemes \cite{PappasGC10,PappasITW11,PappasJCN2016}.

In Section \ref{sec:model} we describe the system model and in Section \ref{sec:general_region} we calculate the general stability region and obtain the convexity conditions. In Sections \ref{sec:IAN_region} and \ref{sec:SIC_region} we consider the case of IAN and SIC, respectively. In Section \ref{sec:SICIAN_region} we consider the case where one receiver treats interference as noise and the other uses SIC. In Section \ref{sec:region_RA}, we derive the stability region in the general case under random access and in Section \ref{sec:closure} we consider the closure of the stability region. Numerical evaluation of the analytical results is provided in Section \ref{sec:results} and Section \ref{sec:conclusions} concludes the paper.

\section{System Model} \label{sec:model}

We consider a two-user interference channel, as depicted in Fig.~\ref{fig:model}, in which each source $S_i,i=1,2$ intends to communicate with and to convey information to its respective destination $D_{i},i=1,2$. 
This setting models two independent wireless systems operating in the same spectrum; communication between transmitters and receivers takes place simultaneously when they have packets to transmit. The packet arrival processes at $S_1$ and $S_2$ are assumed to be independent and stationary with mean rates $\lambda_1$ and $\lambda_2$, respectively. Transmitter $S_i$ has an infinite capacity queue to store incoming packets and $Q_i$ denotes the size in number packets of the $i$-th queue. The packets are assumed to have the same size.

\begin{figure}[t]
\centering
\includegraphics[scale=1.5]{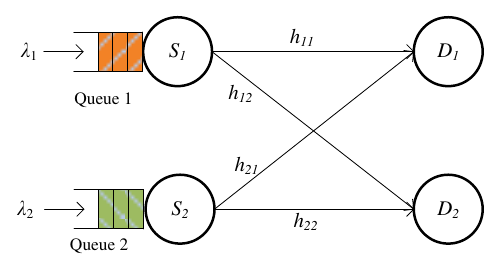}
\caption{Two-user interference channel with bursty arrivals.}
\label{fig:model}
\end{figure}

Time is assumed to be slotted and each source transmits a packet in a timeslot if its queue is not empty; otherwise it remains silent. \footnote{In Section \ref{sec:region_RA}, we consider the case of random access.} Encoding and decoding are performed on a per-slot basis and the transmission of one packet requires one timeslot. We assume the availability of instantaneous and error-free acknowledgments (ACKs); this is a simplifying but widely-used assumption in this kind of studies. If a packet fails to be decoded correctly by its intended destination, then the packet will remain in its queue and will be retransmitted in the next time slot.  

We assume a frequency-flat block fading channel, i.e., the fading coefficients $h_{ij}$ follow the same distribution and remain constant during one timeslot, but change independently from one timeslot to another. Unless otherwise stated, we consider Rayleigh fading, i.e., channel coefficients follow a circularly symmetric complex Gaussian distribution with zero mean and unit variance. The noise is assumed to be additive white Gaussian with zero mean and unit variance. The transmission power of source $S_i$ is denoted by $p_i$, and $r_{ij}$ denotes the distance between transmitter $S_i$ and receiver $D_j$. In addition to small-scale fading, we consider large-scale pathloss attenuation, which is denoted by the non-increasing function of the distance $r$, $\ell(r)$.

A packet is considered to be received/decoded correctly if the received signal-to-interference-plus-noise ratio (SINR) is above a predefined threshold $\gamma$. Let $\mathcal{D}^\mathcal{T}_{i}$ denote the event that the $i$-th receiver is able to decode the packet transmitted from the $i$-th source given a set of active transmitters denoted by $\mathcal{T}$, i.e., $\mathcal{D}^{\{ 1,2 \} }_{1}$ denotes the event that the first receiver can decode the information from the first source when both transmitters are active ($\mathcal{T} = \{1,2\}$). This event is defined as
\begin{equation}
\mathcal{D}^{\{ i \} }_{\mathcal{T}} \triangleq \left\lbrace \mathrm{SINR}_i \geq \gamma_i\right\rbrace.
\end{equation}
The events $\mathcal{D}^{\{ i,j \} }_{i}$ (both sources are active) are defined based on the received SINR, which in turn depends on the specific interference treatment on each receiver; specific expressions are provided in subsequent sections.
When only $S_i$ is active the event $\mathcal{D}^{\{ i \} }_{i}$ is defined as
\begin{equation}
\mathcal{D}^{\{ i \} }_{i} \triangleq \left\lbrace \mathrm{SNR}_i \geq \gamma_i\right\rbrace,
\end{equation}
which denotes that the received signal-to-noise ratio (SNR) of the $i$-th receiver is above a certain threshold $\gamma_i$.

\subsection{Stability Criterion}

We adopt the following definition of queue stability~\cite{Szpankowski:stability}:

\begin{definition}
Denote by $Q_i^t$ the length of queue $i$ at the beginning of time slot $t$. The queue is said to be \emph{stable} if
$\displaystyle \lim_{t \rightarrow \infty} {Pr}[Q_i^t < {x}] = F(x)$ and $\displaystyle \lim_{ {x} \rightarrow \infty} F(x) = 1$.
\end{definition}

Although we do not make explicit use of this definition, we use its corollary consequence which is Loynes' theorem~\cite{b:Loynes}. This theorem states that if the arrival and service processes of a queue are strictly jointly stationary and the average arrival rate is less than the average service rate, then the queue is stable. If the average arrival rate is greater than the average service rate, then the queue is unstable and the value of $Q_i^t$ approaches infinity almost surely. 

The stability region of the system is defined as the set of arrival rate vectors $\boldsymbol{\lambda}=(\lambda_1, \lambda_2)$ for which the queues in the system are stable.

\section{Stability Region: General Case} \label{sec:general_region}

In this section we provide the stability region in a parametric form without considering any specific transmission and interference management technique at the receivers. The main result of this section is oblivious to the details of how successful reception is achieved. It is just based on the values of the success probabilities, which depend on several factors including power, rate, link distance, encoding/decoding algorithms. The general conditions for the stability region of the two-user IC are particularized in subsequent sections.

The service rates for the sources are given by
\begin{equation} \label{eq:mu_1}
\mu_1 = \mathrm{Pr} [Q_2 > 0] \mathrm{Pr}\left(\mathcal{D}^{\{1,2\}}_{1}\right) + \mathrm{Pr} [Q_2 = 0] \mathrm{Pr}\left(\mathcal{D}^{\{1\}}_{1}\right),
\end{equation}

\begin{equation} \label{eq:mu_2}
\mu_2 = \mathrm{Pr} [Q_1 > 0] \mathrm{Pr}\left(\mathcal{D}^{\{1,2\}}_{2}\right) + \mathrm{Pr} [Q_1 = 0] \mathrm{Pr}\left(\mathcal{D}^{\{2\}}_{2}\right).
\end{equation}

Since the average service rate of each queue depends on the queue size of the other queues, it cannot be computed directly. Therefore, we apply the stochastic dominance technique \cite{rao:stability}: we construct hypothetical dominant systems, in which one of the sources transmits dummy packets when its packet queue is empty, while the other transmits according to its traffic.

\subsection{The first dominant system}

We consider the first dominant system, in which $S_1$ transmits dummy packets whenever its queue is empty, while $S_2$ behaves in the same way as in the original system. All other assumptions remain unaltered in the dominant system.

From Loyne's criterion~\cite{b:Loynes} it is known that the queue at the second source is stable if and only if $\lambda_2 < \mu_2$. Therefore, the stability condition is given by

\begin{equation}  \label{eq:stable2_dom1}
\lambda_2 < \mu_2 = \mathrm{Pr}\left(\mathcal{D}^{\{1,2\}}_{2}\right).
\end{equation}

From Little's theorem~\cite{b:Bertsekas}, the probability that the queue of the second transmitter is empty is given by

\begin{equation} \label{eq:q2nonempty_dom1}
\mathrm{Pr}[Q_2 > 0]=\frac{\lambda_2}{\mathrm{Pr}\left(\mathcal{D}^{\{1,2\}}_{2}\right)}.
\end{equation}

Substituting (\ref{eq:q2nonempty_dom1}) into (\ref{eq:mu_1}), we have that the service rate for the first source is given by

\begin{equation} \label{eq:mu1_dom1}
\mu_1 = \mathrm{Pr}\left(\mathcal{D}^{\{1\}}_{1} \right) - \frac{\lambda_2 \mathrm{Pr}\left(\mathcal{D}^{\{1\}}_{1}\right)}{\mathrm{Pr}\left(\mathcal{D}^{\{1,2\}}_{2}\right)} + \frac{\lambda_2\mathrm{Pr}\left(\mathcal{D}^{\{1,2\}}_{1}\right)}{\mathrm{Pr}\left(\mathcal{D}^{\{1,2\}}_{2}\right)}.
\end{equation}

The queue at the first source is stable if and only if $\lambda_1 < \mu_1$; hence the stability condition is given by
\begin{equation} \label{eq:stable1_dom1}
\lambda_1 < \mathrm{Pr}\left(\mathcal{D}^{\{1\}}_{1}\right) - \frac{\lambda_2 \mathrm{Pr}\left(\mathcal{D}^{\{1\}}_{1}\right)}{\mathrm{Pr}\left(\mathcal{D}^{\{1,2\}}_{2}\right)} + \frac{\lambda_2\mathrm{Pr}\left(\mathcal{D}^{\{1,2\}}_{1}\right)}{\mathrm{Pr}\left(\mathcal{D}^{\{1,2\}}_{2}\right)}.
\end{equation}

The stability region $\mathcal{R}_1$, obtained by the first dominant system and conditions (\ref{eq:stable1_dom1}) and (\ref{eq:stable2_dom1}), is given by (\ref{eq:R_1}) (on top of the next page).

\begin{figure*}[!t]
\begin{equation} \label{eq:R_1}
\mathcal{R}_1 = \left\lbrace (\lambda_{1},\lambda_{2}): \frac{\lambda_1}{\mathrm{Pr}\left(\mathcal{D}^{\{1\}}_{1}\right)} + \frac{\left[\mathrm{Pr}\left(\mathcal{D}^{\{1\}}_{1}\right) - \mathrm{Pr}\left(\mathcal{D}^{\{1,2\}}_{1}\right) \right] \lambda_2}{\mathrm{Pr}\left(\mathcal{D}^{\{1\}}_{1}\right)\mathrm{Pr}\left(\mathcal{D}^{\{1,2\}}_{2}\right)}   <  1,
\lambda_2 < \mathrm{Pr}\left(\mathcal{D}^{\{1,2\}}_{2}\right)  \right\rbrace
\end{equation}
\begin{equation} \label{eq:R_2}
\mathcal{R}_2 = \left\lbrace (\lambda_{1},\lambda_{2}): \frac{\lambda_2}{\mathrm{Pr}\left(\mathcal{D}^{\{2\}}_{2}\right)} + \frac{\left[\mathrm{Pr}\left(\mathcal{D}^{\{2\}}_{2}\right) - \mathrm{Pr}\left(\mathcal{D}^{\{1,2\}}_{2}\right) \right] \lambda_1}{\mathrm{Pr}\left(\mathcal{D}^{\{2\}}_{2}\right)\mathrm{Pr}\left(\mathcal{D}^{\{1,2\}}_{1}\right)}   <  1,
\lambda_1 < \mathrm{Pr}\left(\mathcal{D}^{\{1,2\}}_{1}\right)  \right\rbrace
\end{equation}
\end{figure*}

\subsection{The second dominant system}

In the second dominant system, source $S_2$ transmits dummy packets when its queue is empty and all other assumptions remain unaltered.
Similarly to the first dominant system and using Loyne's criterion, the stability condition is given by
\begin{equation}  \label{eq:stable1_dom2}
\lambda_1 < \mu_1 = \mathrm{Pr}\left(\mathcal{D}^{\{1,2\}}_{1}\right),
\end{equation}
since the queue at the first source is stable if and only if $\lambda_1 < \mu_1$.

The probability that the queue of the first user is empty is given by
\begin{equation} \label{eq:q1nonempty_dom2}
\mathrm{Pr}[Q_1 > 0]=\frac{\lambda_1}{\mathrm{Pr}\left(\mathcal{D}^{\{1,2\}}_{1}\right)}.
\end{equation}

Therefore, substituting (\ref{eq:q1nonempty_dom2}) into (\ref{eq:mu_2}) and given that the second queue is stable if and only if $\lambda_2 < \mu_2$, the stability condition is given by
\begin{equation} \label{eq:stable2_dom2}
\lambda_2 < \mu_2 = \mathrm{Pr}\left(\mathcal{D}^{\{2\}}_{2} \right) - \frac{\lambda_1 \mathrm{Pr}\left(\mathcal{D}^{\{2\}}_{2}\right)}{\mathrm{Pr}\left(\mathcal{D}^{\{1,2\}}_{1}\right)} + \frac{\lambda_1\mathrm{Pr}\left(\mathcal{D}^{\{1,2\}}_{2}\right)}{\mathrm{Pr}\left(\mathcal{D}^{\{1,2\}}_{1}\right)}.
\end{equation}

The stability region $\mathcal{R}_2$, obtained by the second dominant system and conditions (\ref{eq:stable2_dom2}) and (\ref{eq:stable1_dom2}), is given on the top of this page by (\ref{eq:R_2}).

An important observation made in \cite{rao:stability} is that the stability conditions obtained by the stochastic dominance technique are not only sufficient but also necessary conditions for the stability of the original system. The \emph{indistinguishability} argument~\cite{rao:stability} applies to our problem as well. Based on the construction of the dominant system, it is easy to see that the queues of the dominant system are always larger in size than those of the original system, provided they are both initialized to the same value. Therefore, given $\lambda_{2}<\mu_{2}$, if for some $\lambda_{1}$, the queue at $S_1$ is stable in the dominant system, then the corresponding queue in the original system must be stable. Conversely, if for some $\lambda_{1}$ in the dominant system, the queue at node $S_1$ saturates, then it will not transmit dummy packets, and as long as $S_1$ has a packet to transmit, the behavior of the dominant system is identical to that of the original system because dummy packet transmissions are eliminated as we approach the stability boundary. Therefore, the original and the dominant systems are indistinguishable at the boundary points.

The stability region is given by $\mathcal{R} = \mathcal{R}_1 \bigcup \mathcal{R}_2$ and is depicted in Fig.~\ref{fig:region}. The stability region is a two-dimensional {\emph{convex polyhedron}} when the following condition holds:

\begin{equation}\label{eq:convexity_general}
\frac{\mathrm{Pr}\left(\mathcal{D}^{\{1,2\}}_{1}\right)}{\mathrm{Pr}\left(\mathcal{D}^{\{1\}}_{1}\right)}+
\frac{\mathrm{Pr}\left(\mathcal{D}^{\{1,2\}}_{2}\right)}{\mathrm{Pr}\left(\mathcal{D}^{\{2\}}_{2}\right)} \geq 1.
\end{equation}

\begin{figure}[t]
\centering
\includegraphics[scale=0.6]{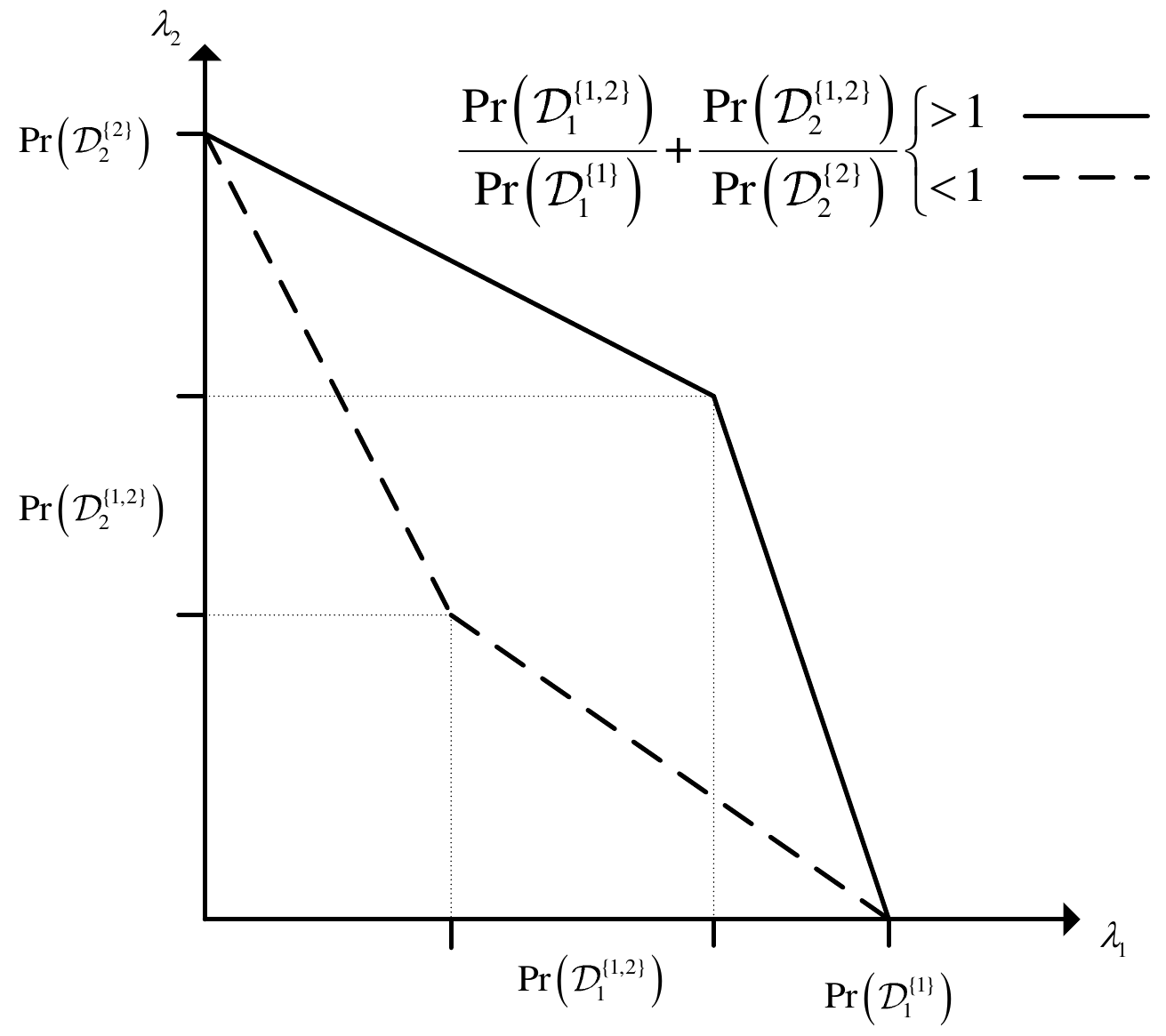}
\caption{The stability region for the general case.}
\label{fig:region}
\end{figure}

When (\ref{eq:convexity_general}) holds with equality, the stability region is a triangle and coincides with the case of time sharing. Convexity is an important property since it corresponds to the case when parallel concurrent transmissions are preferable to a time-sharing scheme. Additionally, convexity of the stability region implies that if two rate pairs are stable, then any rate pair lying on the line segment joining those two rate pairs is also stable.

The aforementioned result on the stability region of IC holds for any interference management technique. In the following sections, we particularize the stability conditions for the following interference harnessing cases: treating interference as noise and successive interference cancellation.

\section{IAN at Both Receivers} \label{sec:IAN_region}
In this section, we consider the case where both receivers decode their individual messages by treating interference from unintended sources as noise. Treating interference as noise is a widely used interference management strategy that has several practical and theoretical appeals. From a practical perspective, IAN involves the use of point-to-point channel codes, which are well understood and relatively easy to implement. In addition to its low complexity, IAN is shown to be robust to channel uncertainty and can be applied with coarse channel state information at the transmitter (CSIT). From a theoretical perspective, IAN is shown to be optimal for weak interference. Note that although IAN is simple technique, the rate region under IAN is general non convex.

When both transmitters are active, the event $\mathcal{D}^{\{ i,j \} }_{i}$ is given by
\begin{equation}
\mathcal{D}^{\{ i,j \} }_{i} \triangleq \left\lbrace \mathrm{SINR}_i \geq \gamma_i \right\rbrace \triangleq \left\lbrace \frac{|h_{ii}|^2 \ell(r_{ii}) p_i}{1+|h_{ji}|^2 \ell(r_{ji}) p_j} \geq \gamma_i \right\rbrace.
\end{equation}

Before proceeding further, we provide the following result for the success probability between the $i$-th transmitter and the $i$-th receiver when all the sources are active. 
\begin{proposition}
	\label{Ps_gen}
Denote the interference by $I = |h_{ji}|^2 \ell(r_{ji}) p_j$ and suppose $F_{|h_{ii}|^2}$ takes the form 
\begin{equation}
F_{|h_{ii}|^2}(x)= 1 - e^{-x}\sum_{k\in\mathcal{K}}a_{k} x^k
\end{equation}
for a finite set $\mathcal{K}$ and for valid distribution functions. Assuming that $|h_{ii}|^2$ and $|h_{ji}|^2$ are independent, then
\begin{equation}
\label{eq:PSIR} 
\mathrm{Pr}\left(\mathcal{D}^{\{i,j\}}_{i}\right) =  \sum_{k\in\mathcal{K}} \sum_{m=0}^{k}\left[a_{k}\binom{k}{m}(-1)^{2k-m}\phi^k e^{-\phi}\frac{d^m}{d\phi^m} \mathcal{L}_{I}(\phi)\right]
\end{equation}	
where $\phi=\frac{\gamma_i}{p_i \ell(r_{ii})}$ and $\mathcal{L}_{X}(s) = \mathbb{E}(e^{-sX})$ is the Laplace transform.
\end{proposition}

\begin{proof}
The success probability is given (denoting $W=1+I$)
\begin{eqnarray*}
\mathrm{Pr}\left(\mathcal{D}^{\{i,j\}}_{i}\right)=  \mathrm{Pr}\left(|h_{ii}|^2 \geq \phi W\right)  = \sum_{k\in\mathcal{K}}a_{k}\phi^k \mathbb{E}_I\left(e^{-\phi W}W^k\right)\\
\stackrel{(a)}{=} \sum_{k\in\mathcal{K}}a_{k}(-\phi)^k \frac{d^k}{d\phi^k} \mathcal{L}_{W}(\phi) = \sum_{k\in\mathcal{K}}a_{k}(-\phi)^k \frac{d^k}{d\phi^k} (e^{-\phi}\mathcal{L}_{I}(\phi))
\end{eqnarray*}             
where (a) uses the Laplace transform property $y^kf(y)\longleftrightarrow (-1)^k\frac{\mathrm{d}^k}{\mathrm{d}s^k}\mathcal{L}[f(y)](s)$, with $f(y)$ denoting the probability density function of $I$. The final result is obtained using the general Leibniz rule for the higher order derivatives of products of two functions and after some algebraic manipulations.
\end{proof}

The above result enables us to perform outage analysis for a large of received signal and interference fading distributions, including more importantly a number of multi-antenna techniques.

For Rayleigh fading and single-antenna nodes, the probability that the channel between the $i$-th transmitter and the $i$-th receiver is not in outage when all the sources are active is given by 
\begin{eqnarray} \label{eq:SINR_IAN}
&& \mathrm{Pr}\left(\mathcal{D}^{\{i,j\}}_{i}\right) =\mathrm{Pr} \left\lbrace \mathrm{SINR}_i \geq \gamma_i \right\rbrace = e^{-\phi}\mathcal{L}_{I}(\phi)\nonumber \\ && = \exp \left(- \frac{\gamma_i}{p_{i}\ell(r_{ii})} \right) \left[1+\gamma_i \frac{p_{j}\ell(r_{ji})}{p_{i}\ell(r_{ii})} \right]^{-1} ,\text{ for }i=1,2
\end{eqnarray}
That is a direct application of Proposition \ref{Ps_gen} for $k=0$, $a_k=1$, and $|h_{ji}|^2 \sim \exp(1)$.

The probability that the link $ii$ is not in outage when only $S_i$ is active is given by
\begin{eqnarray} \label{eq:SNR}
\mathrm{Pr}\left(\mathcal{D}^{\{i\}}_{i}\right) & = & \mathrm{Pr} \left\lbrace \mathrm{SNR}_i \geq \gamma_i \right\rbrace \nonumber \\
& = & \mathrm{Pr}\left\lbrace |h_{ii}|^2 \ell(r_{ii}) p_i \geq \gamma_i\right\rbrace = \exp \left(-\frac{\gamma_i}{p_{i}\ell(r_{ii})}\right).
\end{eqnarray}

Substituting (\ref{eq:SNR}), (\ref{eq:SINR_IAN}) to (\ref{eq:R_1}), (\ref{eq:R_2}), we obtain the stability sub-regions (\ref{eq:R_1_IAN}) and (\ref{eq:R_2_IAN}) respectively, given on the top of next page.
\begin{figure*}
\begin{equation} \label{eq:R_1_IAN}
\mathcal{R}^{\mathrm{IAN}}_1= \left\lbrace  (\lambda_{1},\lambda_{2}): \frac{\lambda_1}{\exp\left(- \frac{\gamma_1}{p_{1}\ell(r_{11})}\right)} + \frac{\gamma_1 \frac{p_2\ell(r_{21})}{p_1\ell(r_{11})} + \gamma_1 \gamma_2 \frac{\ell(r_{12})\ell(r_{21})}{\ell(r_{11})\ell(r_{22})}}{\exp\left(-\frac{\gamma_2}{p_{2}\ell(r_{22})} \right)} \lambda_2 < 1, \lambda_2 < \frac{\exp \left(-\frac{\gamma_2}{p_{2}\ell(r_{22})} \right)}{ \left[1+\gamma_2 \frac{p_{1}\ell(r_{12})}{p_{2}\ell(r_{22})}\right]}  \right\rbrace
\end{equation}
\begin{equation} \label{eq:R_2_IAN}
\mathcal{R}^{\mathrm{IAN}}_2= \left\lbrace  (\lambda_{1},\lambda_{2}): \frac{\lambda_2}{\exp \left(-\frac{\gamma_2}{p_{2}\ell(r_{22})}\right)} + \frac{\gamma_2 \frac{p_1\ell(r_{12})}{p_2\ell(r_{22})} + \gamma_1 \gamma_2 \frac{\ell(r_{12})\ell(r_{21})}{\ell(r_{22})\ell(r_{11})}}{\exp\left(-\frac{\gamma_1}{p_{1}\ell(r_{11})} \right)} \lambda_1 < 1, \lambda_1 < \frac{\exp \left(-\frac{\gamma_1}{p_{1}\ell(r_{11})}\right)}{\left[1+\gamma_1 \frac{p_{2}\ell(r_{21})}{p_{1}\ell(r_{11})}\right]}  \right\rbrace
\end{equation}
\begin{equation}\label{eq:R_1_IAN_C}
\mathcal{R}^{\mathrm{IAN}}_1= \left\lbrace  (\lambda_{1},\lambda_{2}): \frac{\lambda_1}{\mathrm{Pr} \left\lbrace \mathrm{SNR}_1 \geq \gamma_1 \right\rbrace
} + \frac{\mathrm{Pr} \left\lbrace \mathrm{SNR}_1 \geq \gamma_1 \right\rbrace - \mathrm{Pr} \left\lbrace \mathrm{SINR}_1 \geq \gamma_1 \right\rbrace}{\mathrm{Pr} \left\lbrace \mathrm{SNR}_1 \geq \gamma_1 \right\rbrace \mathrm{Pr} \left\lbrace \mathrm{SINR}_2 \geq \gamma_2 \right\rbrace} \lambda_2 < 1, \lambda_2 < \mathrm{Pr} \left\lbrace \mathrm{SINR}_2 \geq \gamma_2 \right\rbrace  \right\rbrace
\end{equation}
\begin{equation}\label{eq:R_2_IAN_C}
\mathcal{R}^{\mathrm{IAN}}_2= \left\lbrace  (\lambda_{1},\lambda_{2}): \frac{\lambda_2}{\mathrm{Pr} \left\lbrace \mathrm{SNR}_2 \geq \gamma_2 \right\rbrace
} + \frac{\mathrm{Pr} \left\lbrace \mathrm{SNR}_2 \geq \gamma_2 \right\rbrace - \mathrm{Pr} \left\lbrace \mathrm{SINR}_2 \geq \gamma_2 \right\rbrace}{\mathrm{Pr} \left\lbrace \mathrm{SNR}_2 \geq \gamma_2 \right\rbrace \mathrm{Pr} \left\lbrace \mathrm{SINR}_1 \geq \gamma_1 \right\rbrace} \lambda_1 < 1, \lambda_1 < \mathrm{Pr} \left\lbrace \mathrm{SINR}_1 \geq \gamma_1 \right\rbrace  \right\rbrace
\end{equation}
\end{figure*}
The $\mathcal{R}^{\mathrm{IAN}}_1$ and $\mathcal{R}^{\mathrm{IAN}}_2$ can be presented in a more compact form given in (\ref{eq:R_1_IAN_C}) and (\ref{eq:R_2_IAN_C}), respectively.

The stability region when both receivers can decode their messages by treating interference as noise is $\mathcal{R}^{\mathrm{IAN}} = \mathcal{R}^{\mathrm{IAN}}_1 \cup \mathcal{R}^{\mathrm{IAN}}_2$.

Substituting (\ref{eq:SNR}) and (\ref{eq:SINR_IAN}) to (\ref{eq:convexity_general}) we have that $\mathcal{R}^{\mathrm{IAN}}$ is a convex set when $\frac{\mathrm{Pr}\left\lbrace\mathrm{SINR}_1 \geq \gamma_1 \right\rbrace}{\mathrm{Pr}\left\lbrace\mathrm{SNR}_1 \geq \gamma_1\right\rbrace} + \frac{\mathrm{Pr}\left\lbrace\mathrm{SINR}_2 \geq \gamma_2 \right\rbrace}{\mathrm{Pr}\left\lbrace\mathrm{SNR}_2 \geq \gamma_2 \right\rbrace} \geq 1$, which leads to the condition

\begin{equation} \label{eq:convexity_IAN}
\gamma_1 \gamma_2 \leq \frac{\ell(r_{11})\ell(r_{22})}{\ell(r_{12})\ell(r_{21})}.
\end{equation}

The $\mathcal{R}^{\mathrm{IAN}}$ when (\ref{eq:convexity_IAN}) holds, is depicted in Fig.~\ref{fig:SIC_vs_IAN}.

\begin{figure}[t]
\centering
\includegraphics[scale=0.65]{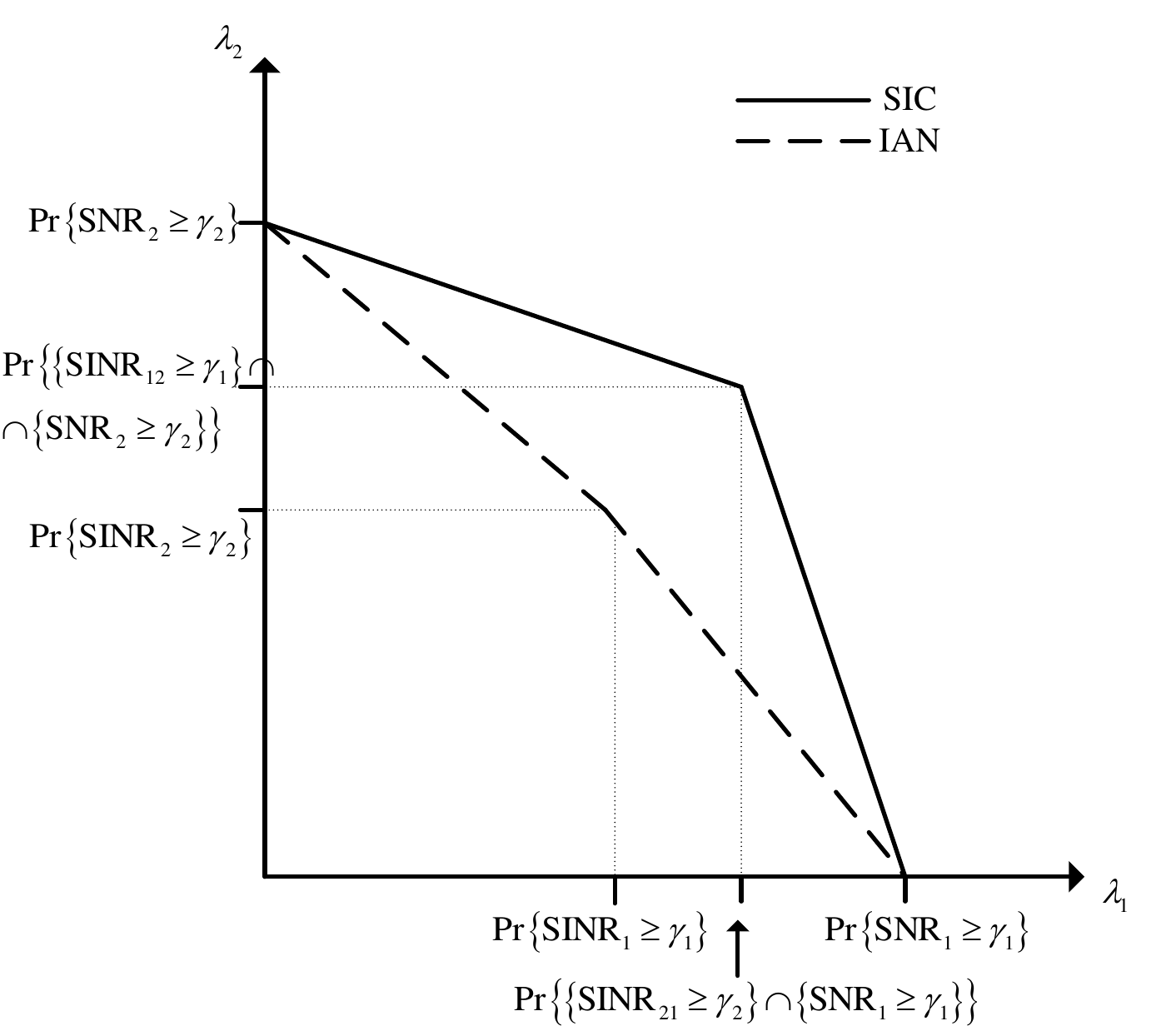}
\caption{The stability region for the case where both receivers apply IAN and SIC. The condition (\ref{eq:SIC_optimal_condition0}) holds for both receivers.}
\label{fig:SIC_vs_IAN}
\end{figure}

\subsection{The effect of Multiple Antennas at the Transmitters}
We consider here that the transmitters have $M$ transmit antennas each, and the destinations have a single receive antenna. The problem setup constitutes a multiple-input single-output (MISO) IC. We focus on single-stream transmission (scalar coding followed by beamforming) that is simple yet optimal under certain conditions (e.g., under perfect CSIT and IAN). 

We consider two simple and widely used  strategies for choosing beamformers.

\subsubsection{Maximum ratio transmission beamforming} 
The first is the maximum ratio transmission beamforming, which maximizes the mutual information (rate) of each transmitter. Note that this beamforming strategy has been shown to achieve a unique, pure Nash equilibrium, i.e., no single transmitter can improve its rate by switching to a different beamformer.
In that case, the unit-norm beamforming vectors are $\mathbf{w}_i^{{\rm MRT}} = \frac{\mathbf{h}_{ii}^{H}}{\left\|\mathbf{h}_{ii}\right\|}$. The received signal is therefore $|\mathbf{h}_{ii}^H\mathbf{w}_i^{{\rm MRT}}|^2 = \left\|\mathbf{h}_{ii}\right\|^2$ and follows a gamma distribution with shape parameter $M$ and scale parameter 1, i.e., $ \left\|\mathbf{h}_{ii}\right\|^2 \sim \Gamma(M,1)$. The interference fading distribution is $|\mathbf{h}_{ji}^H\mathbf{w}_j^{{\rm MRT}}|^2$ and is exponentially distributed. The probability of successful transmission is given by Proposition \ref{Ps_gen} for $k=M$ and $a_k=1/k!$ as follows:
\begin{equation} \label{eq:SINR_miso}
\mathrm{Pr}\left(\mathcal{D}^{\{i,j\}}_{i}\right) = \sum_{k=0}^{M-1} \sum_{m=0}^{k}\binom{k}{m}(-1)^{2k-m}\frac{\phi^ke^{-\phi}}{k!} \frac{d^m}{d\phi^m} \frac{1}{1+\phi p_{j}\ell(r_{ji})}.
\end{equation}

The probability that the link $ii$ is not in outage when only $S_i$ is active is given by
\begin{equation} \label{eq:SNR_miso}
\mathrm{Pr}\left(\mathcal{D}^{\{i\}}_{i}\right) = \sum_{k=0}^{M-1} \left(\frac{\gamma_i}{p_i\ell(r_{ii})}\right)^k\cdot \frac{1}{k!}\cdot \exp\left(-\frac{\gamma_i}{p_i\ell(r_{ii})}\right).
\end{equation}

\subsubsection{Zero forcing beamforming} 
The zero-forcing strategy aims at eliminating co-channel interference in order to maximize the mutual information (rate) of other users. In other words, $S_i$ employs a transmit strategy that creates no interference at all for $D_j$, and vice versa. The unit-norm beamforming vector is given by
\begin{equation} 
\mathbf{w}_i^{{\rm ZF}} = \frac{\Pi_{\mathbf{h}_{ij}}^{\perp} \mathbf{h}_{ii}}{\left\|\Pi_{\mathbf{h}_{ij}}^{\perp} \mathbf{h}_{ii}\right\|}
\end{equation}
where $\Pi_{\mathbf{A}}^{\perp}$ denotes projection onto the orthogonal complement of the space of $\mathbf{A}$. For the received signal we have $|\mathbf{h}_{ii}^H\mathbf{w}_i^{{\rm ZF}}|^2 \sim \Gamma(M-1,1)$ and interference is eliminated.

The success probability between the $i$-th transmitter and the $i$-th receiver, when both sources are active is given by
\begin{eqnarray}\label{eq:SINR_ZF}
&&\mathrm{Pr}\left(\mathcal{D}^{\{i,j\}}_{i}\right) =    \mathrm{Pr}\left\lbrace |\mathbf{h}_{ii}^H\mathbf{w}_i^{{\rm ZF}}|^2 p_i \ell(r_{ii}) \geq \gamma_i \right\rbrace \nonumber \\
&& =   \sum_{k=0}^{M-2} \left(\frac{\gamma_i}{p_i\ell(r_{ii})}\right)^k\cdot \frac{1}{k!}\cdot \exp\left(-\frac{\gamma_i}{p_i\ell(r_{ii})}\right).
\end{eqnarray}

Note that a linear combination of MRT and ZF beamformers, which attempts to find the best tradeoff of the two, is another interesting beamforming scheme that is shown to provide a complete description of the Pareto boundary for the achievable rate of MISO two-user interference channel \cite{MISO_Pareto}.

\section{SIC at both receivers} \label{sec:SIC_region}
If the received interference is strong, successive interference cancellation may be used at the destination. 
In this section we consider the case where both receivers employ successive interference cancellation when both transmitters are active. Assuming that the receiver $D_j$ knows the codebook of $S_i$, it can perform SIC: $D_j$ can decode the message sent by $S_i$ first and then subtract it from the received signal, thus eliminating the interference term before decoding its own message. 
Note that despite its appealing optimality in the strong interference regime, SIC is associated with several challenging implementation issues, including timing- and frequency synchronization, assumption on knowing the modulation and coding of each other.

Receiver $D_i$ is able to decode the transmitted packet by $S_i$ when both sources are active, if the following conditions are satisfied

\begin{align}
\gamma_j \leq \frac{|h_{ji}|^2 \ell(r_{ji}) p_j}{1+|h_{ii}|^2 \ell(r_{ii}) p_i} \triangleq \mathrm{SINR}_{ji}\text{ and } \gamma_i \leq \mathrm{SNR}_i.
\end{align}

Thus, the event $\mathcal{D}^{\{i,j\}}_{i}$ is given by $\mathcal{D}^{\{i,j\}}_{i} = \left\lbrace \mathrm{SINR}_{ji} \geq \gamma_j \right\rbrace \cap \left\lbrace \mathrm{SNR}_i \geq \gamma_i \right\rbrace$.
The following lemma provides the probability that $D_i$ can decode the transmitted information from $S_i$ given that both sources all active.

\begin{lemma}
The success probability between the $i$-th transmitter and the $i$-th receiver, when both sources are active and the $i$-th receiver performs SIC is given by
\begin{equation}\label{eq:SINR_SIC}
\begin{aligned}
\mathrm{Pr}\left(\mathcal{D}^{\{i,j\}}_{i}\right) =\mathrm{Pr}\left\lbrace \mathrm{SINR}_{ji} \geq \gamma_j \cap  \mathrm{SNR}_i \geq \gamma_i  \right\rbrace \\
= \exp \left(-\frac{\gamma_i}{p_{i}\ell(r_{ii})} \right) \cdot \frac{\exp \left(-\frac{\gamma_j(1+\gamma_i)}{p_{j}\ell(r_{ji})} \right)}  {1+\gamma_j \frac{p_{i}\ell(r_{ii})}{p_{j}\ell(r_{ji})}}.
\end{aligned}
\end{equation}
\end{lemma}

\begin{proof}
Without loss of generality, we provide the proof for the success probability of the first link. Transmission to the first receiver, when both sources are active, is successful if $\left\lbrace \mathrm{SINR}_{ji} \geq \gamma_j \right\rbrace \cap \left\lbrace \mathrm{SNR}_i \geq \gamma_i \right\rbrace$, which is equivalent to

\begin{equation}
\left\lbrace \frac{|h_{21}|^2 \ell(r_{21}) p_2}{1+|h_{11}|^2 \ell(r_{11}) p_1} \geq \gamma_2 \right\rbrace \cap \left\lbrace p_1 |h_{11}|^2 \ell(r_{11}) \geq \gamma_1 \right\rbrace,
\end{equation}
which can written as
\begin{equation}
\left\lbrace |h_{21}|^2 \geq \frac{\gamma_2+ \gamma_2|h_{11}|^2 \ell(r_{11}) p_1}{p_2 \ell(r_{21})} \right\rbrace \cap \left\lbrace  |h_{11}|^2 \geq \frac{\gamma_1}{p_1 \ell(r_{11})} \right\rbrace.
\end{equation}

The success probability can be expressed as
\begin{eqnarray}
&&\mathrm{Pr} \left[\left\lbrace |h_{21}|^2 \geq \frac{\gamma_2+ \gamma_2|h_{11}|^2 \ell(r_{11}) p_1}{p_2 \ell(r_{21})} \right\rbrace \cap \left\lbrace  |h_{11}|^2 \geq \frac{\gamma_1}{p_1 \ell(r_{11})} \right\rbrace  \right] \nonumber \\ 
&& = \int_{\frac{\gamma_1}{p_1 \ell(r_{11})}}^{\infty} \! \mathrm{Pr} \left[ |h_{21}|^2 \geq \frac{\gamma_2+ \gamma_2 \ell(r_{11}) p_1 x}{\ell(r_{21}) p_2} \;\middle|\; |h_{11}|^2=x \right] f_{|h_{11}|^2}(x)\mathrm{d}x.
\end{eqnarray}

Then, we have that the success probability can be written as
\begin{equation}
\mathrm{Pr}\left(\mathcal{D}^{\{1,2\}}_{1}\right) = \int_{\frac{\gamma_1}{p_1 \ell(r_{11})}}^{\infty} \! \overbar{F}_{|h_{21}|^2} \left( \frac{\gamma_2+ \gamma_2 \ell(r_{11}) p_1 x}{\ell(r_{21}) p_2}  \right)  f_{|h_{11}|^2}(x) \mathrm{d}x.
\end{equation}
where $\overbar{F}(x) = 1 - F(x)$ is the complementary cdf.

For Rayleigh fading, we have $f_{|h_{11}|^2}(x)= e^{-x}$ and $F_{|h_{11}|^2} (x) = 1-e^{-x}$, hence the success probability is given by
\begin{equation}
\mathrm{Pr}\left(\mathcal{D}^{\{1,2\}}_{1}\right) = \exp \left(- \frac{\gamma_1}{p_{1}\ell(r_{11})} \right) \frac{ \exp \left(- \frac{\gamma_2(1+\gamma_1)}{p_{2} \ell(r_{21})} \right)}  {\left[1+\gamma_2 \frac{p_{1}\ell(r_{11})}{p_{2}\ell(r_{21})} \right]}.
\end{equation}

\end{proof}

The probability $\mathrm{Pr}(\mathcal{D}^{\{i\}}_{i}) = \mathrm{Pr}(\mathrm{SNR}_i \geq \gamma_i)$ is given by (\ref{eq:SNR}).
Substituting (\ref{eq:SNR}) and (\ref{eq:SINR_SIC}) to (\ref{eq:R_1}) and (\ref{eq:R_2}) we obtain that the subregions are (\ref{eq:R_1_SIC_C}) and (\ref{eq:R_2_SIC_C}).

\begin{figure*}
\begin{equation}\label{eq:R_1_SIC_C}
\begin{aligned}
\mathcal{R}^{\mathrm{SIC}}_1= \left\lbrace  (\lambda_{1},\lambda_{2}): \frac{\lambda_1}{\mathrm{Pr} \left\lbrace \mathrm{SNR}_1 \geq \gamma_1 \right\rbrace
} + \frac{1-\mathrm{Pr}\left\lbrace \mathrm{SINR}_{21} \geq \gamma_2 \mid \mathrm{SNR}_1 \geq \gamma_1 \right\rbrace}{\mathrm{Pr}\left\lbrace \mathrm{SINR}_{12} \geq \gamma_1 \cap \mathrm{SNR}_2 \geq \gamma_2 \right\rbrace} \lambda_2 < 1, \right. \\
 \left. \lambda_2 < \mathrm{Pr}\left\lbrace \mathrm{SINR}_{12} \geq \gamma_1 \cap  \mathrm{SNR}_2 \geq \gamma_2 \right\rbrace  \right\rbrace
\end{aligned}
\end{equation}
\begin{equation}\label{eq:R_2_SIC_C}
\begin{aligned}
\mathcal{R}^{\mathrm{SIC}}_2= \left\lbrace  (\lambda_{1},\lambda_{2}): \frac{\lambda_2}{\mathrm{Pr} \left\lbrace \mathrm{SNR}_2 \geq \gamma_2 \right\rbrace
} + \frac{1-\mathrm{Pr}\left\lbrace \mathrm{SINR}_{12} \geq \gamma_1  \mid \mathrm{SNR}_2 \geq \gamma_2  \right\rbrace}{\mathrm{Pr}\left\lbrace \mathrm{SINR}_{21} \geq \gamma_2  \cap  \mathrm{SNR}_1 \geq \gamma_1  \right\rbrace} \lambda_1 < 1, \right. \\
 \left. \lambda_1 < \mathrm{Pr}\left\lbrace \mathrm{SINR}_{21} \geq \gamma_2  \cap  \mathrm{SNR}_1 \geq \gamma_1  \right\rbrace  \right\rbrace
\end{aligned}
\end{equation}
\end{figure*}

The stability region $\mathcal{R}^{\mathrm{SIC}} = \mathcal{R}^{\mathrm{SIC}}_1 \cup \mathcal{R}^{\mathrm{SIC}}_2 $ is a convex set if 
\begin{align*}
\mathrm{Pr}\left\lbrace \mathrm{SINR}_{21} \geq \gamma_2 \mid \mathrm{SNR}_1 \geq \gamma_1 \ \right\rbrace + \\ \mathrm{Pr}\left\lbrace \mathrm{SINR}_{12} \geq \gamma_1 \mid \mathrm{SNR}_2 \geq \gamma_2 \ \right\rbrace \geq 1.
\end{align*}

The $\mathcal{R}^{\mathrm{SIC}}$ is depicted in Fig.~\ref{fig:SIC_vs_IAN} when is a convex set.

\subsection{SIC vs. IAN}
A natural question is under which conditions SIC is better than IAN in the sense that $R^{\mathrm{IAN}} \subset R^{\mathrm{SIC}}$.
Comparing $R^{\mathrm{IAN}}$ and $R^{\mathrm{SIC}}$, we observe that SIC provides better performance when the following condition is satisfied for both receivers (see also Fig.~\ref{fig:SIC_vs_IAN}):

\begin{equation}\label{eq:SIC_optimal_condition0}
\mathrm{Pr} \left\lbrace \mathrm{SINR}_{i} \geq \gamma_i \right\rbrace <\mathrm{Pr}\left\lbrace \mathrm{SINR}_{ji} \geq \gamma_j  \cap \mathrm{SNR}_i \geq \gamma_i \ \right\rbrace ,
\end{equation}
which leads to the following condition after substitution:
\begin{equation} \label{eq:SIC_optimal_condition}
\frac{1+\gamma_j \frac{p_{i}\ell(r_{ii})}{p_{j}\ell(r_{ji})}}{1+\gamma_i \frac{p_{j}\ell(r_{ji})}{p_{i}\ell(r_{ii})} } < \exp \left(- \frac{\gamma_j (1+\gamma_i)}{p_{j}\ell(r_{ji})} \right).
\end{equation}

If the condition (\ref{eq:SIC_optimal_condition}) is not satisfied at both receivers, then IAN provides superior performance as compared to SIC. In the case that the condition is not satisfied at $D_i$ but holds for $D_j$, then IAN should be used for $D_i$ and SIC for $D_j$. The stability region for the latter case is investigated in the next section.

\section{SIC at $D_1$ - IAN at $D_2$} \label{sec:SICIAN_region}
Without loss of generality, we consider that the first receiver decodes the interference using SIC (condition (\ref{eq:SIC_optimal_condition}) holds for $D_1$), whereas the second receiver employs IAN, i.e., inequality (\ref{eq:SIC_optimal_condition}) holds with the opposite direction.

For the first destination, which decodes the transmitted message applying SIC, the probabilty of successful event $\mathrm{Pr}\left(\mathcal{D}^{\{1,2\}}_{1}\right) = \mathrm{Pr}\left\lbrace \mathrm{SNR}_1 \geq \gamma_1 \cap \mathrm{SINR}_{21} \geq \gamma_2 \right\rbrace,$ is given by (\ref{eq:SINR_SIC}).

For the second destination, the probability $\mathrm{Pr}\left(\mathcal{D}^{\{1,2\}}_{2}\right) =\mathrm{Pr} \left\lbrace  \mathrm{SINR}_2 \geq \gamma_2 \right\rbrace$ is given by (\ref{eq:SINR_IAN}).
Note that when only $i$-th source transmits, then we need that the $SNR_i$ to be greater than threshold $\gamma_i$.

The stability region is given by $\mathcal{R}^{\mathrm{SIC-IAN}}= \mathcal{R}^{\mathrm{SIC-IAN}}_1 \cup \mathcal{R}^{\mathrm{SIC-IAN}}_2 $ and is shown in Fig.~\ref{fig:SIC_IAN}. The subregions (\ref{eq:R_1_IANSIC}) and (\ref{eq:R_2_IANSIC}) are obtained respectively by substituting (\ref{eq:SINR_IAN}) and (\ref{eq:SINR_SIC}) into (\ref{eq:R_1}) and (\ref{eq:R_2}) for the first and section destination respectively.

\begin{figure*}
\begin{equation} \label{eq:R_1_IANSIC}
\begin{aligned}
\mathcal{R}^{\mathrm{SIC-IAN}}_1= \left\{  (\lambda_{1},\lambda_{2}):
\frac{\lambda_1}{\mathrm{Pr} \left\lbrace\mathrm{SNR}_{1} \geq \gamma_1 \right\rbrace}+ \frac{1-\mathrm{Pr}\left\lbrace  \mathrm{SINR}_{21} \geq \gamma_2 \mid \mathrm{SNR}_1 \geq \gamma_1 \right\rbrace}{\mathrm{Pr}\left\lbrace\mathrm{SINR}_{2} \geq \gamma_2 \right\rbrace}\lambda_2 < 1, \lambda_2 < \mathrm{Pr}\left\lbrace\mathrm{SINR}_{2} \geq \gamma_2 \right\rbrace \right\}
\end{aligned}
\end{equation}
\begin{equation}
\begin{aligned} \label{eq:R_2_IANSIC}
\mathcal{R}^{\mathrm{SIC-IAN}}_2= \left\{  (\lambda_{1},\lambda_{2}): \frac{\lambda_2}{\mathrm{Pr} \left\lbrace\mathrm{SNR}_{2} \geq \gamma_2 \right\rbrace} + \frac{\mathrm{Pr} \left\lbrace\mathrm{SNR}_{2} \geq \gamma_2\right\rbrace-\mathrm{Pr} \left\lbrace\mathrm{SINR}_{2} \geq \gamma_2 \right\rbrace}{\mathrm{Pr}\left\lbrace\mathrm{SNR}_{2} \geq \gamma_2\right\rbrace\ \mathrm{Pr}\left\lbrace  \mathrm{SINR}_{21} \geq \gamma_2  \cap \mathrm{SNR}_1 \geq \gamma_1 \right\rbrace }\lambda_1 < 1, \right. \\
 \left. \lambda_1 < \mathrm{Pr}\left\lbrace \mathrm{SINR}_{21} \geq \gamma_2 \cap \mathrm{SNR}_1 \geq \gamma_1 \right\rbrace \right\}
\end{aligned}
\end{equation}
\end{figure*}

Note that the stability region $\mathcal{R}^{\mathrm{SIC-IAN}}$ is a convex set if $\mathrm{Pr}\left\lbrace \mathrm{SINR}_{21} \geq \gamma_2 \mid \mathrm{SNR}_1 \geq \gamma_1 \right\rbrace + \frac{\mathrm{Pr}\left\lbrace\mathrm{SINR}_2 \geq \gamma_2 \right\rbrace}{\mathrm{Pr}\left\lbrace\mathrm{SNR}_2 \geq \gamma_2 \right\rbrace} \geq 1$.

\begin{figure}[t]
\centering
\includegraphics[scale=0.65]{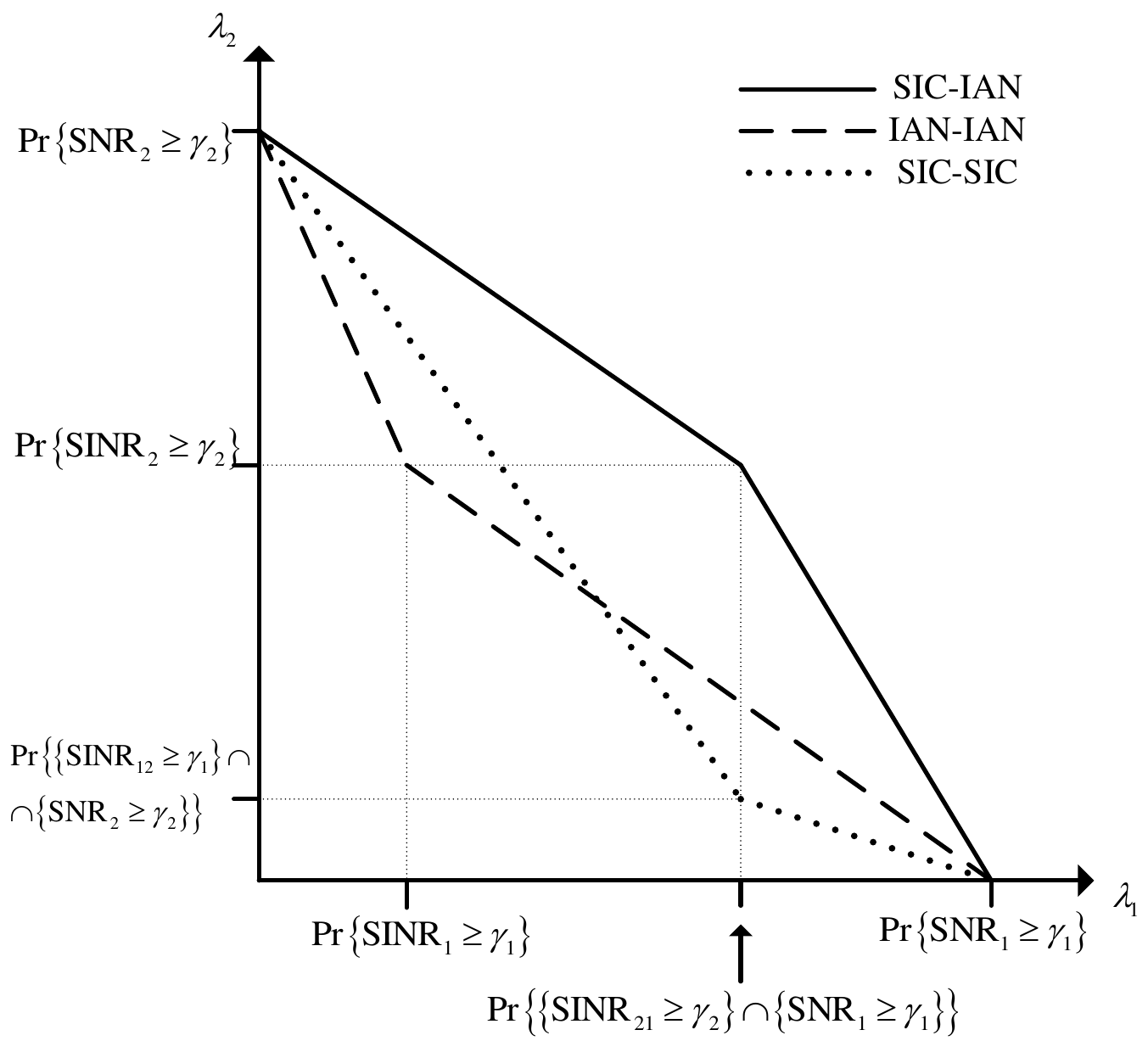}
\caption{Stability region comparison when condition (\ref{eq:SIC_optimal_condition0}) does not hold for both receivers.}
\label{fig:SIC_IAN}
\end{figure}

\section{Stability Region with Random Access} \label{sec:region_RA}

In this section, we extend the previous analysis to the case where the sources transmit in a random access manner. 
Random access is an multiple access protocol whose investigation has been remained vivid for many years. Nevertheless, several fundamental questions remain open even for very simple network configurations. It has recently regained interest with the emergence of massive uncoordinated access and Internet-of-Things (IoT) systems. 

In random access, whenever a source $i = 1,2$, has a non-empty queue transmits a packet with probability $q_i$. Then the service rates for the sources are given by

\begin{align} \label{eq:mu_1_RA}
\mu_1 = q_1 q_2\mathrm{Pr} [Q_2 > 0] \mathrm{Pr}\left(\mathcal{D}^{\{1,2\}}_{1}\right) + \\+ q_1 \left(1-q_2 \mathrm{Pr} [Q_2 = 0] \right) \mathrm{Pr}\left(\mathcal{D}^{\{1\}}_{1}\right),
\end{align}

\begin{align} \label{eq:mu_2_RA}
\mu_2 = q_1 q_2\mathrm{Pr} [Q_1 > 0] \mathrm{Pr}\left(\mathcal{D}^{\{1,2\}}_{2}\right) +\\+ q_2 \left(1-q_1\mathrm{Pr} [Q_1 = 0]\right) \mathrm{Pr}\left(\mathcal{D}^{\{2\}}_{2}\right).
\end{align}

Since the queues are coupled, we proceed similarly to Section \ref{sec:general_region} and apply the methodology of dominant systems. The stability region is obtained as $\mathcal{R}^\mathrm{RA}= \mathcal{R}_1^\mathrm{RA} \cup \mathcal{R}_2^\mathrm{RA}$, where  $\mathcal{R}_1^\mathrm{RA}$ and $\mathcal{R}_2^\mathrm{RA}$ are given by (\ref{eq:R_1_RA}) and (\ref{eq:R_2_RA}), respectively and it is depicted in Fig. \ref{fig:RA}. 

\begin{figure*}[!t]
\begin{equation} \label{eq:R_1_RA}
\begin{aligned}
\mathcal{R}_1^\mathrm{RA} = \left\lbrace (\lambda_{1},\lambda_{2}): \frac{\lambda_1}{q_1 \mathrm{Pr}\left(\mathcal{D}^{\{1\}}_{1}\right)} + \frac{\left[\mathrm{Pr}\left(\mathcal{D}^{\{1\}}_{1}\right) - \mathrm{Pr}\left(\mathcal{D}^{\{1,2\}}_{1}\right) \right] \lambda_2}{\mathrm{Pr}\left(\mathcal{D}^{\{1\}}_{1}\right)\left[(1-q_1)\mathrm{Pr}\left(\mathcal{D}^{\{2\}}_{2}\right)+q_1\mathrm{Pr}\left(\mathcal{D}^{\{1,2\}}_{2}\right) \right]}   <  1, \right. \\
\left.
\lambda_2 < (1-q_1) q_2\mathrm{Pr}\left(\mathcal{D}^{\{2\}}_{2}\right) + q_1 q_2\mathrm{Pr}\left(\mathcal{D}^{\{1,2\}}_{2}\right)  \right\rbrace
\end{aligned}
\end{equation}
\begin{equation} \label{eq:R_2_RA}
\begin{aligned}
\mathcal{R}_2^\mathrm{RA} = \left\lbrace (\lambda_{1},\lambda_{2}): \frac{\lambda_2}{q_2 \mathrm{Pr}\left(\mathcal{D}^{\{2\}}_{2}\right)} + \frac{\left[\mathrm{Pr}\left(\mathcal{D}^{\{2\}}_{2}\right) - \mathrm{Pr}\left(\mathcal{D}^{\{1,2\}}_{2}\right) \right] \lambda_1}{\mathrm{Pr}\left(\mathcal{D}^{\{2\}}_{2}\right)\left[(1-q_2)\mathrm{Pr}\left(\mathcal{D}^{\{1\}}_{1}\right)+q_2\mathrm{Pr}\left(\mathcal{D}^{\{1,2\}}_{1}\right) \right]}   <  1, \right. \\
\left.
\lambda_1 < q_1(1-q_2)\mathrm{Pr}\left(\mathcal{D}^{\{1\}}_{1}\right) + q_1 q_2\mathrm{Pr}\left(\mathcal{D}^{\{1,2\}}_{1}\right)  \right\rbrace
\end{aligned}
\end{equation}
\end{figure*}

\begin{figure}[t]
\centering
\includegraphics[scale=0.65]{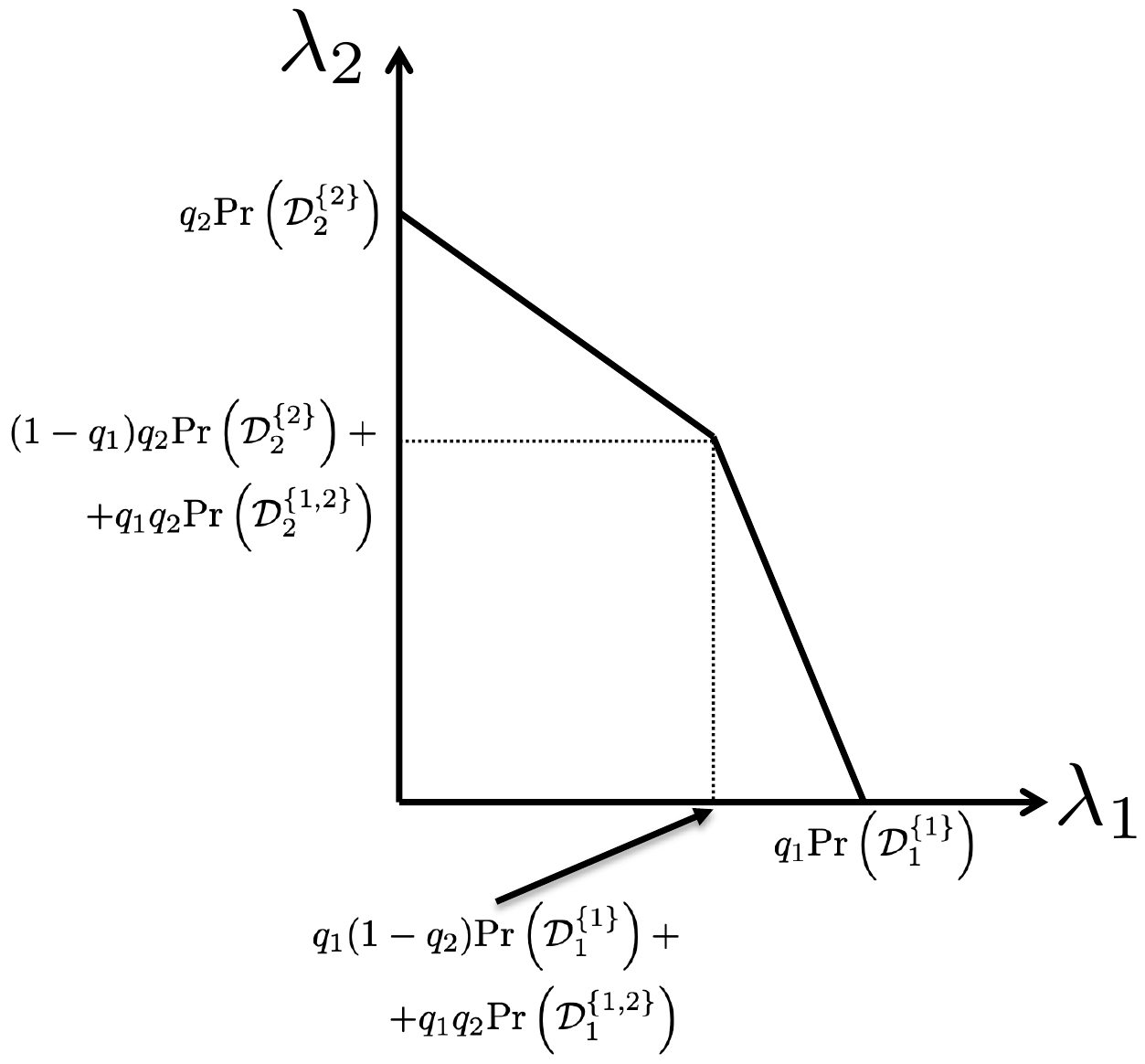}
\caption{Stability region with random access.}
\label{fig:RA}
\end{figure}

In order to compute the success probabilities in the stability region, we can consider the techniques in Sections \ref{sec:IAN_region} - \ref{sec:SICIAN_region}.

\section{Closure of the Stability Region} \label{sec:closure}

In the Sections \ref{sec:general_region} - \ref{sec:SICIAN_region}, we obtained the stability region in terms of success probabilities under the assumption of fixed powers.

If we take the union of these regions over all possible power allocations to the transmitters, we obtain the envelope of the individual stability regions. This corresponds to the closure of the stability region and is defined in (\ref{eq:closure_def}). 

\begin{equation} \label{eq:closure_def}
\mathcal{G}\triangleq \bigcup_{ (p_1, p_2) \in [0,p_\mathrm{max}]^2} \left(  \mathcal{G}_1 (p_1,p_2) \cup \mathcal{G}_2 (p_1,p_2) \right)
\end{equation}

In (\ref{eq:closure_def}), $\mathcal{G}_i (p_1,p_2) \triangleq \mathcal{R}_i$ for $i=1,2$ are obtained Section \ref{sec:general_region} in the general form, and $p_\mathrm{max}$ denotes the maximum power that can be allocated to a transmitter.

In Section \ref{sec:region_RA}, we obtained the stability region when the sources access the medium in a random access manner. We obtained the stability region with fixed transmission probability vectors $(q_{1}, q_{2})$ and with fixed power allocations $(p_{1}, p_{2})$. If we take the union of these regions over all possible transmission probabilities of the users and over all power allocations, we obtain the closure of the stability region defined as

\begin{equation} \label{eq:closure_def_RA_PA}
\begin{array}{rl}
\mathcal{F}\triangleq &\left( \bigcup_{ \boldsymbol{\vec{v}} \in [0,1]^2 \times [0,p_\mathrm{max}]^2} \mathcal{F}_1 (\boldsymbol{\vec{v}}) \right) \\&\bigcup \left( \bigcup_{ \boldsymbol{\vec{v}} \in [0,1]^2 \times [0,p_\mathrm{max}]^2} \mathcal{F}_2 (\boldsymbol{\vec{v}}) \right),
\end{array}
\end{equation}
where $\mathcal{F}_i (\boldsymbol{\vec{v}} ) \triangleq \mathcal{R}_i^\mathrm{RA}$ for $i=1,2$ are obtained in Section \ref{sec:region_RA} and $\boldsymbol{\vec{v}}=(q_{1}, q_{2},p_{1}, p_{2})$.

In general, the analytical derivation of the closure of the stability region is a difficult task. In the next section we will provide numerical evaluation of the closure for several cases.

\section{Numerical Results} \label{sec:results}

In this section, we evaluate numerically the analytical results obtained in the previous sections. We consider a singular power-law pathloss function $\ell(r) = r^{-a}$ where $a$ is the pathloss exponent. Two different topologies of the two-user interference channel are studied : $r_{11}=r_{22}=10$, $r_{12}=r_{21}=5$, and $r_{11}=r_{22}=14$, $r_{12}=15, r_{21}=10$.
We let $a=2$ and $p_{max}=800$. Furthermore, we consider two sets for the SINR thresholds, $\gamma_1=0.5$, $\gamma_2=0.4$ and $\gamma_1=2$, $\gamma_2=1.4$, respectively. 

\subsection{Closure of the stability region over all power allocations} 

\subsubsection{First Topology: $r_{11}=r_{22}=10$, $r_{12}=r_{21}=5$} \label{sec:NumResClosurePower1}

In Fig. \ref{fig:Closure_Case1}  the closure of the stability regions for the case where both the receiver apply IAN, SIC, and for the case where one applies SIC and the second IAN is depicted. The black line in Figs. \ref{fig:Closure_Case1} (a)-(d) corresponds to the stability region for each scheme for the case where both sources transmit with the maximum allowed power $p_{max}$. 

An interesting observation is that for the SIC case, the black line defines the closure of the region which is not the case for IAN. For IAN when the SINR thresholds are less than one, the black line defines the closure for the case where the traffic is low for one transmitter and high for the other. On the other hand, when the SINR thresholds are greater than one, then the black line approaches the closure.

\begin{figure*}[!htbp]
\centering
\subfigure[IAN, $\gamma_1=0.5,\gamma_2=0.4$]{\includegraphics[scale=0.5]{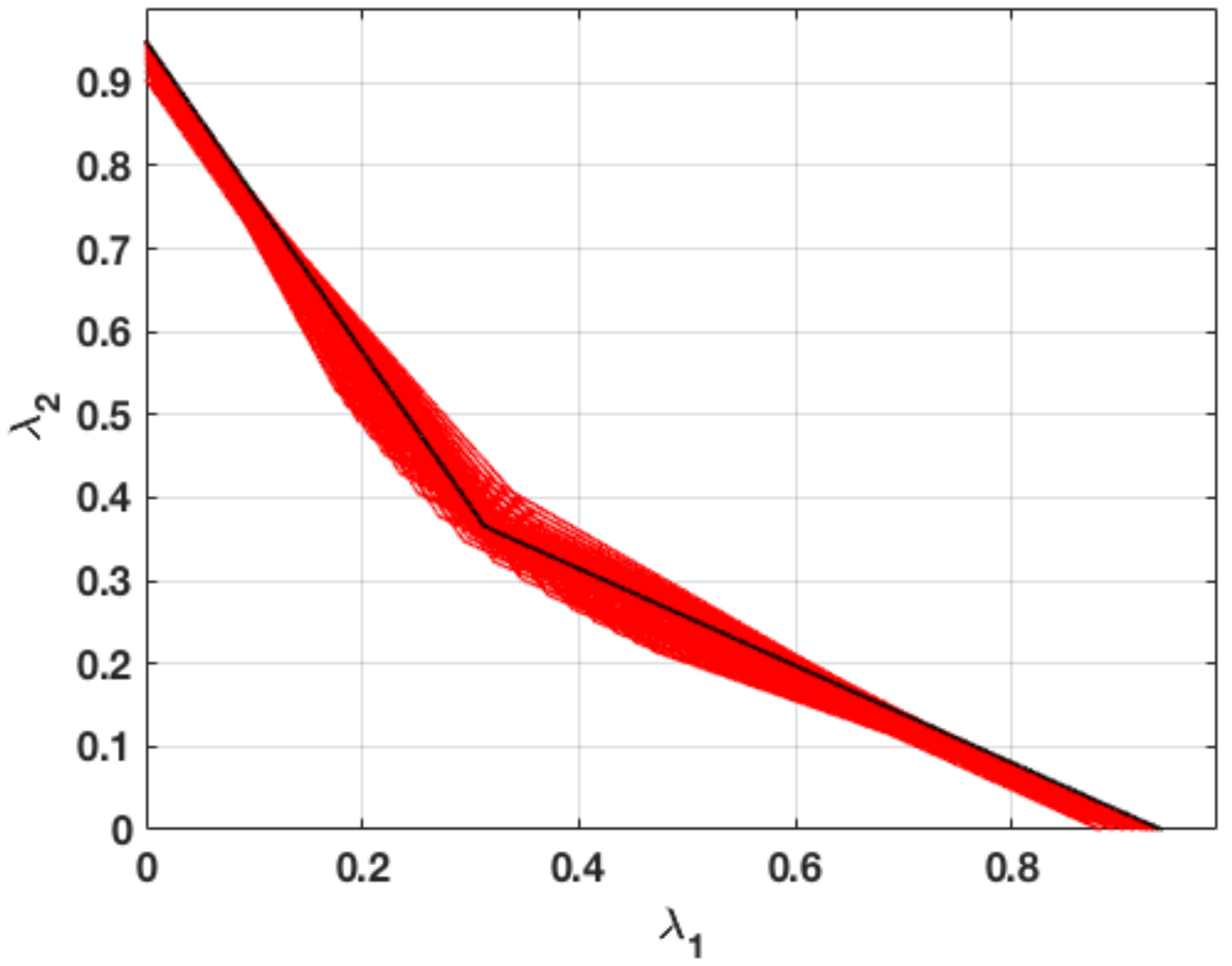}} 
\subfigure[IAN, $\gamma_1=2,\gamma_2=1.4$]{\includegraphics[scale=0.5]{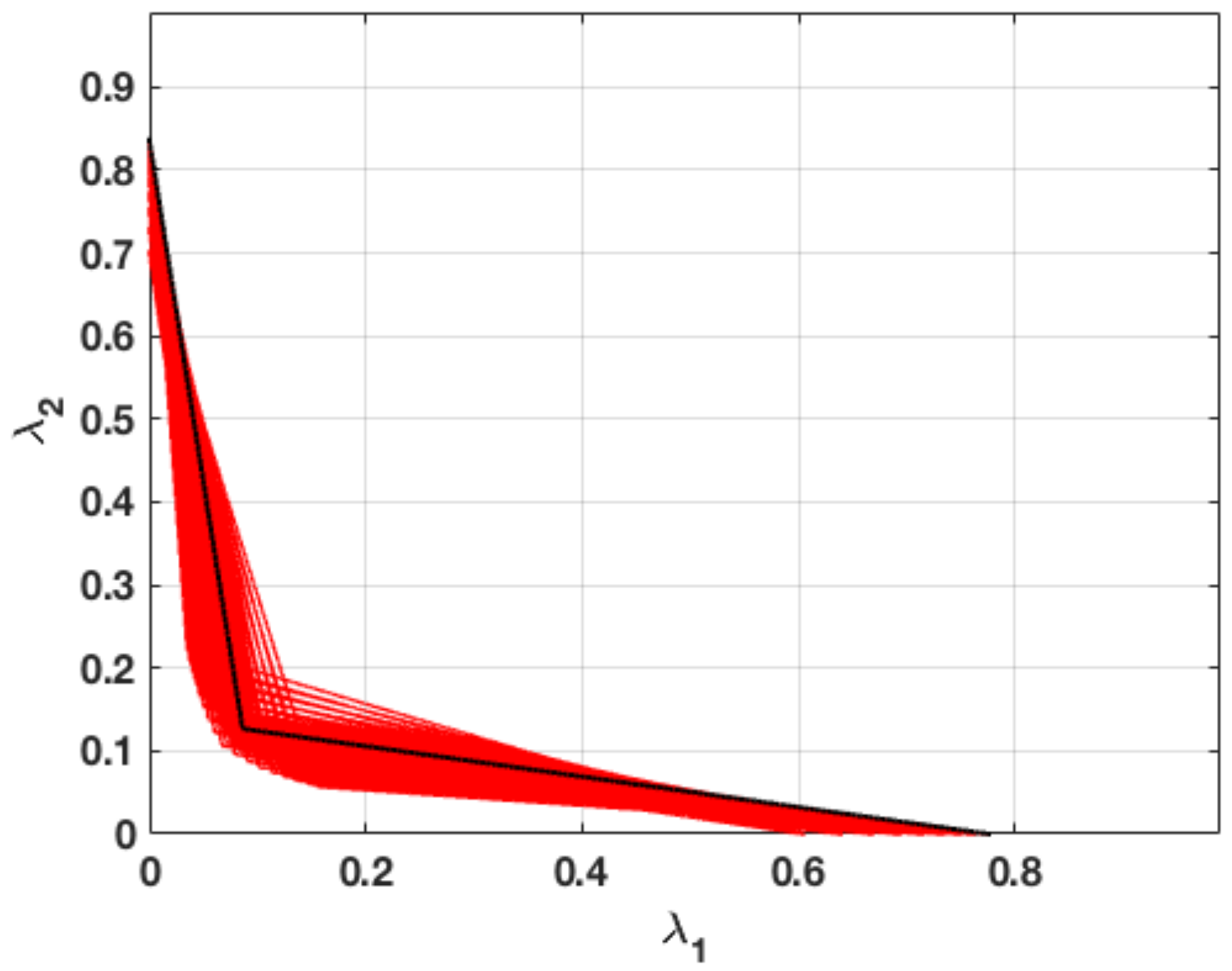}} \\
\subfigure[SIC,$\gamma_1=0.5,\gamma_2=0.4$]{\includegraphics[scale=0.5]{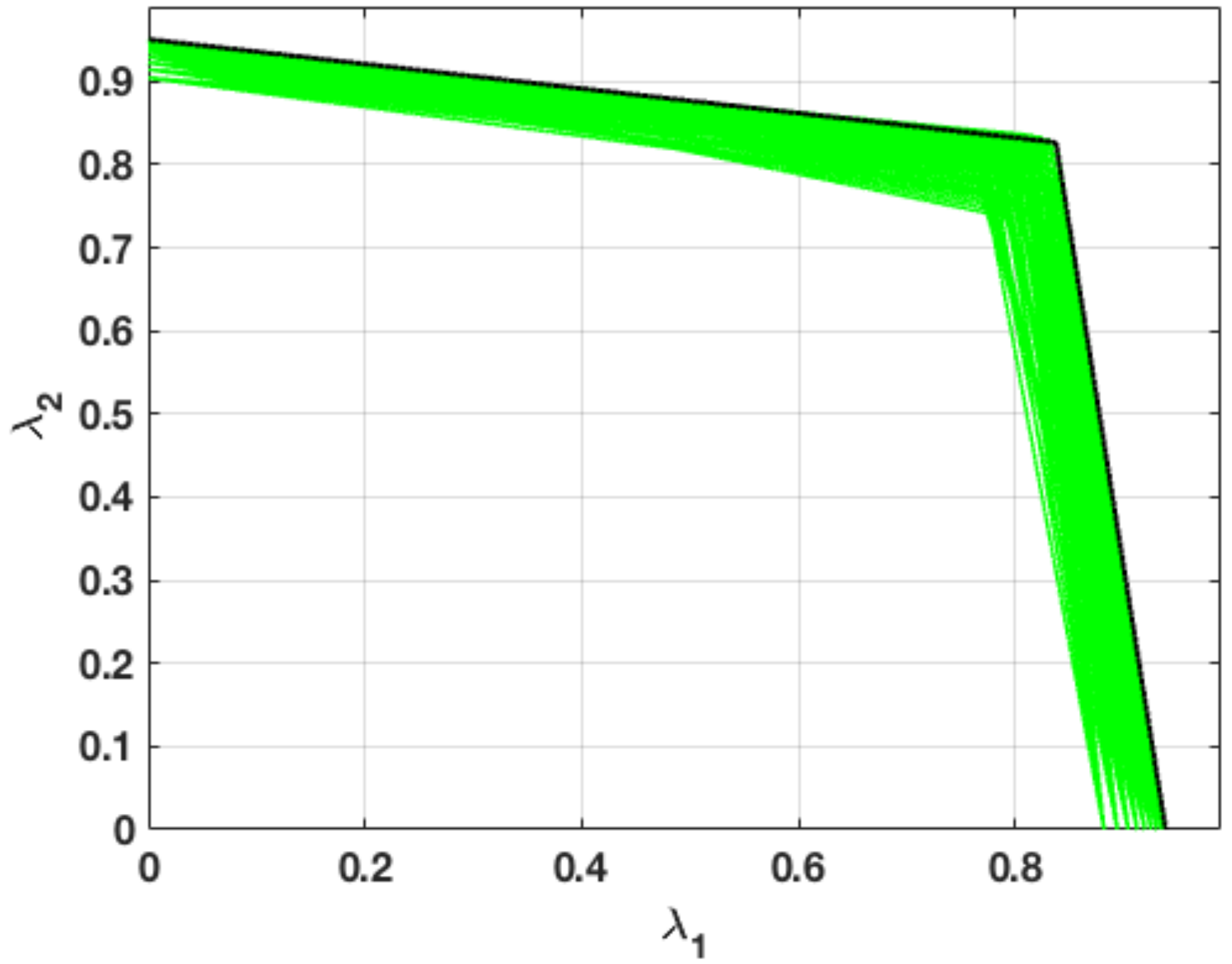}}
\subfigure[SIC,$\gamma_1=2,\gamma_2=1.4$]{\includegraphics[scale=0.5]{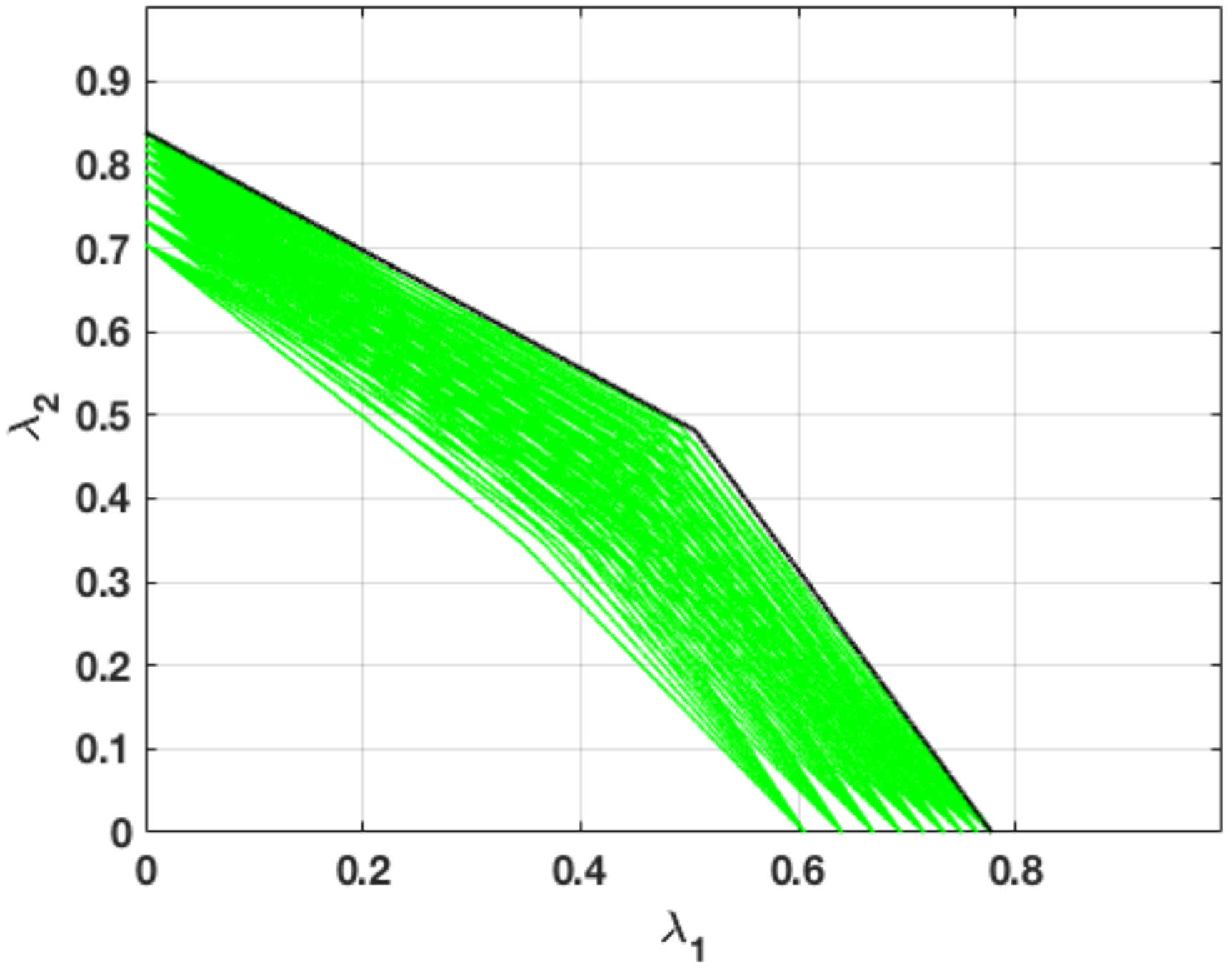}} \\
\subfigure[SIC-IAN,$\gamma_1=0.5,\gamma_2=0.4$]{\includegraphics[scale=0.5]{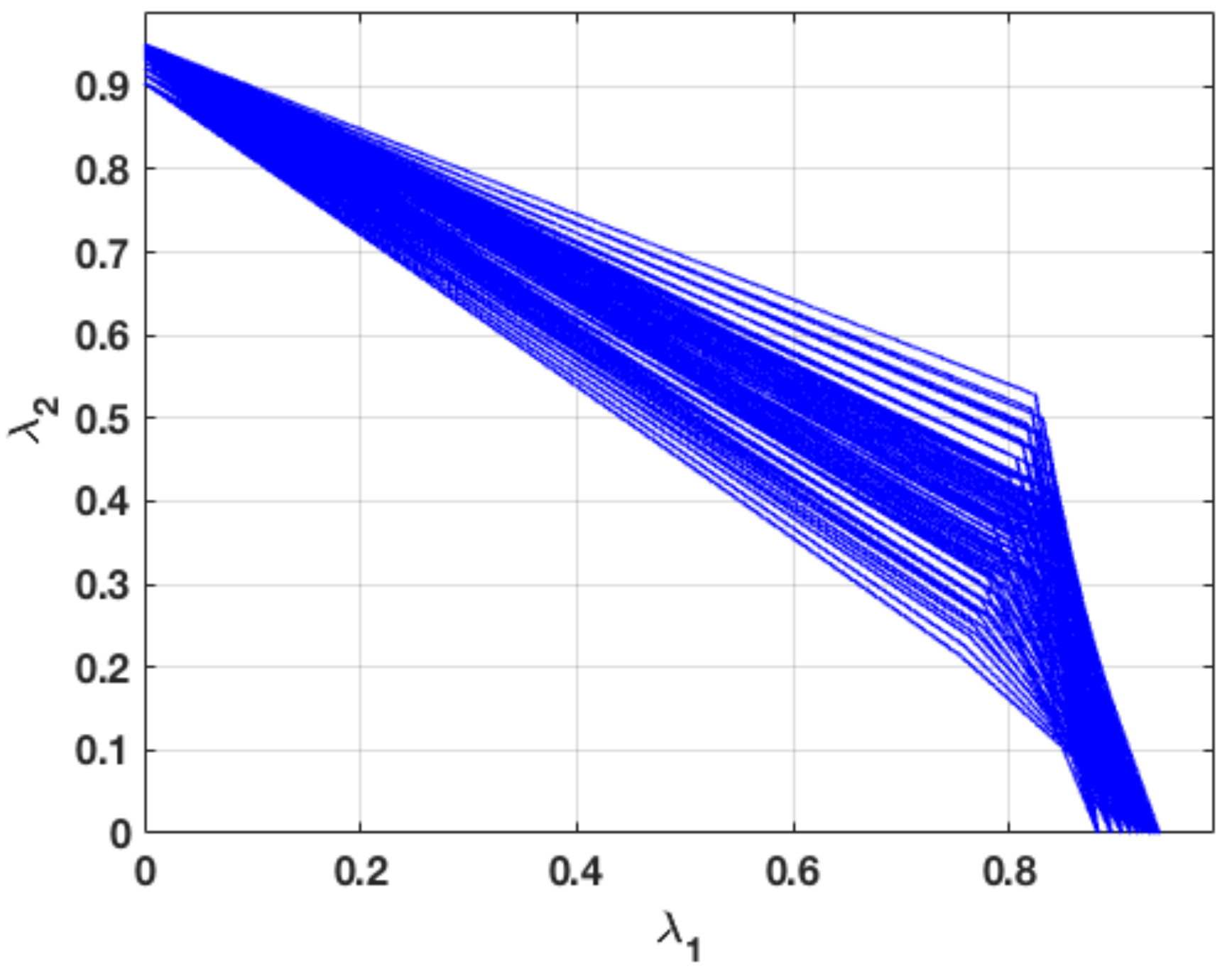}}	
\subfigure[SIC-IAN,$\gamma_1=2,\gamma_2=1.4$]{\includegraphics[scale=0.5]{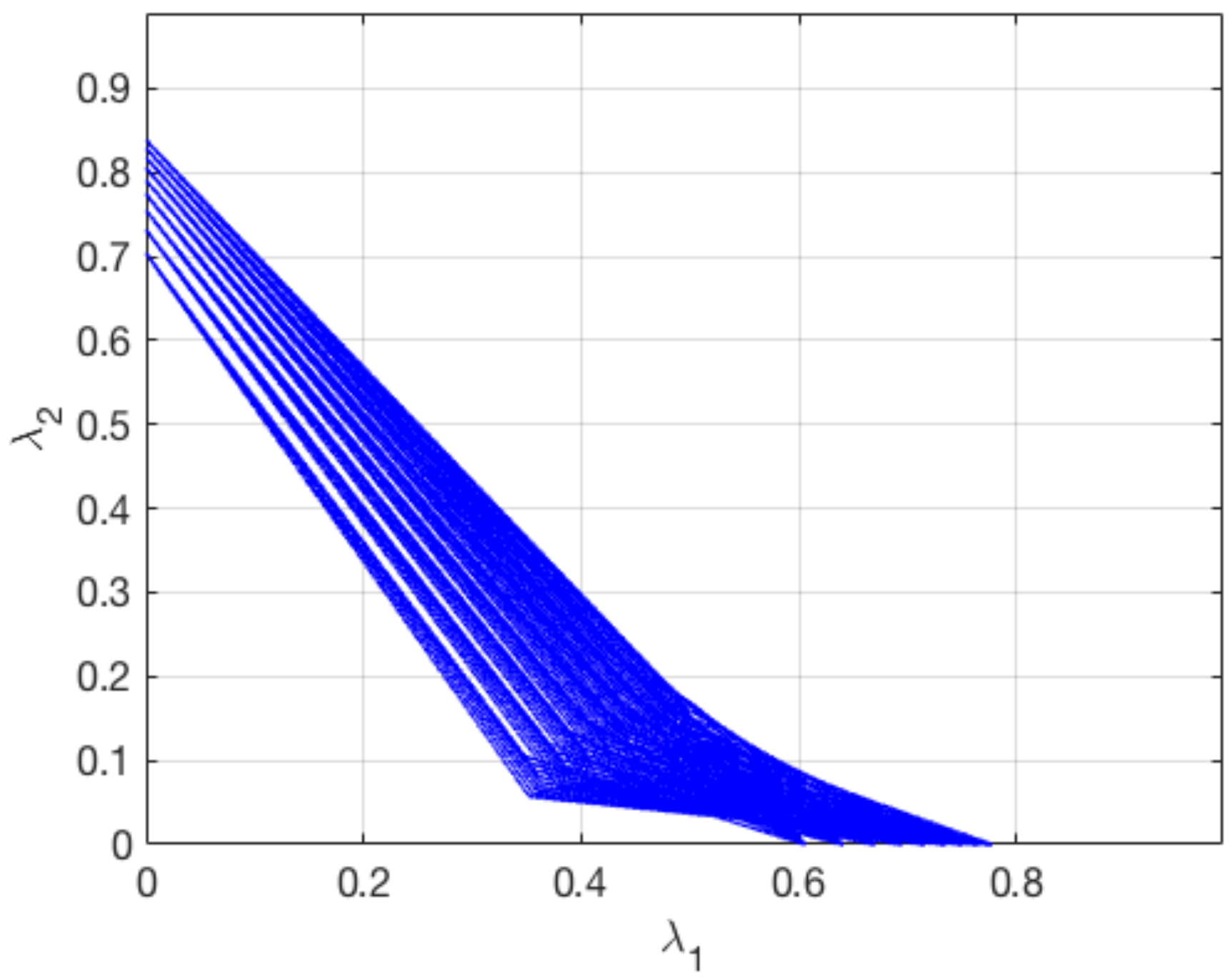}} \\
\caption{The closure of the stability region for the topology where $r_{11}=r_{22}=10, r_{12}=r_{21}=5$.}
\label{fig:Closure_Case1}
\end{figure*}

We observe that SIC has the best performance among the schemes, which is expected, since in this topology, it is easier to decode the received interference first and then decode the intended packet with the SNR criterion.

We observe that for both $\gamma_1=0.5$, $\gamma_2=0.4$ and $\gamma_1=2$, $\gamma_2=1.4$, the closure of the stability region for SIC is a convex set, which establishes its superiority compared to time division. For IAN, we have in both cases a non-convex set, but for the SIC-IAN scheme which is depicted in Fig. \ref{fig:Closure_Case1} (e) and \ref{fig:Closure_Case1} (f), we have a transition from the convex to the non-convex set for high values of the SINR threshold.

\subsubsection{Second Topology: $r_{11}=r_{22}=14$, $r_{12}=15, r_{21}=10$}

In this part, we consider the closure over all power allocations for the three considered schemes and the topology described by $r_{11}=r_{22}=14$, $r_{12}=15, r_{21}=10$. The closures are depicted in Fig. \ref{fig:Closure_Case2} for two cases of SINR thresholds
$\gamma_1=0.5$, $\gamma_2=0.4$ and $\gamma_1=2$, $\gamma_2=1.4$, respectively. 

The black line in Figs. \ref{fig:Closure_Case1} (a)-(d) corresponds to the stability region for each scheme for the case where both sources transmit with the maximum allowed power.

\begin{figure*}[!htbp]
\centering
\subfigure[IAN, $\gamma_1=0.5,\gamma_2=0.4$]{\includegraphics[scale=0.5]{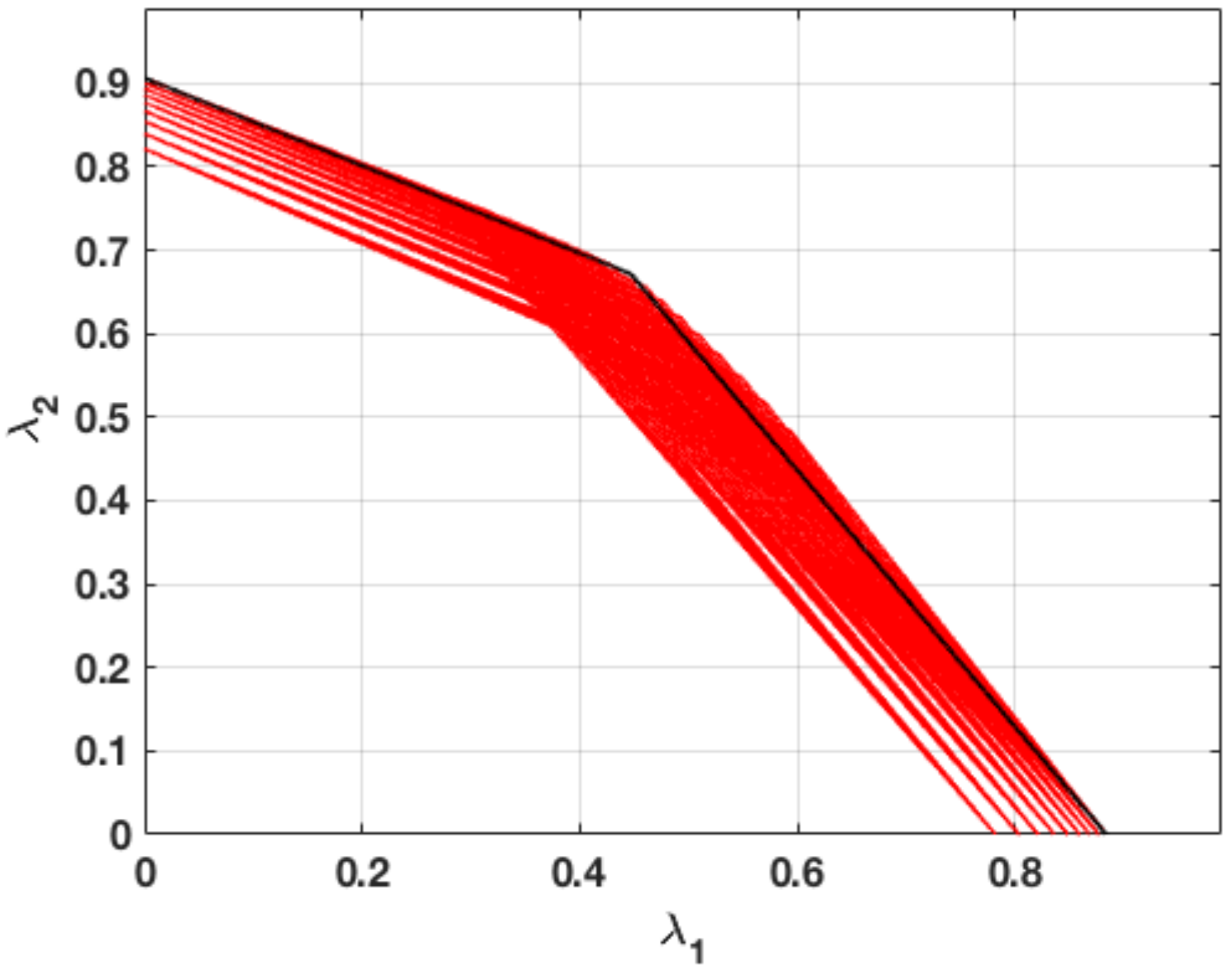}} 
\subfigure[IAN, $\gamma_1=2,\gamma_2=1.4$]{\includegraphics[scale=0.5]{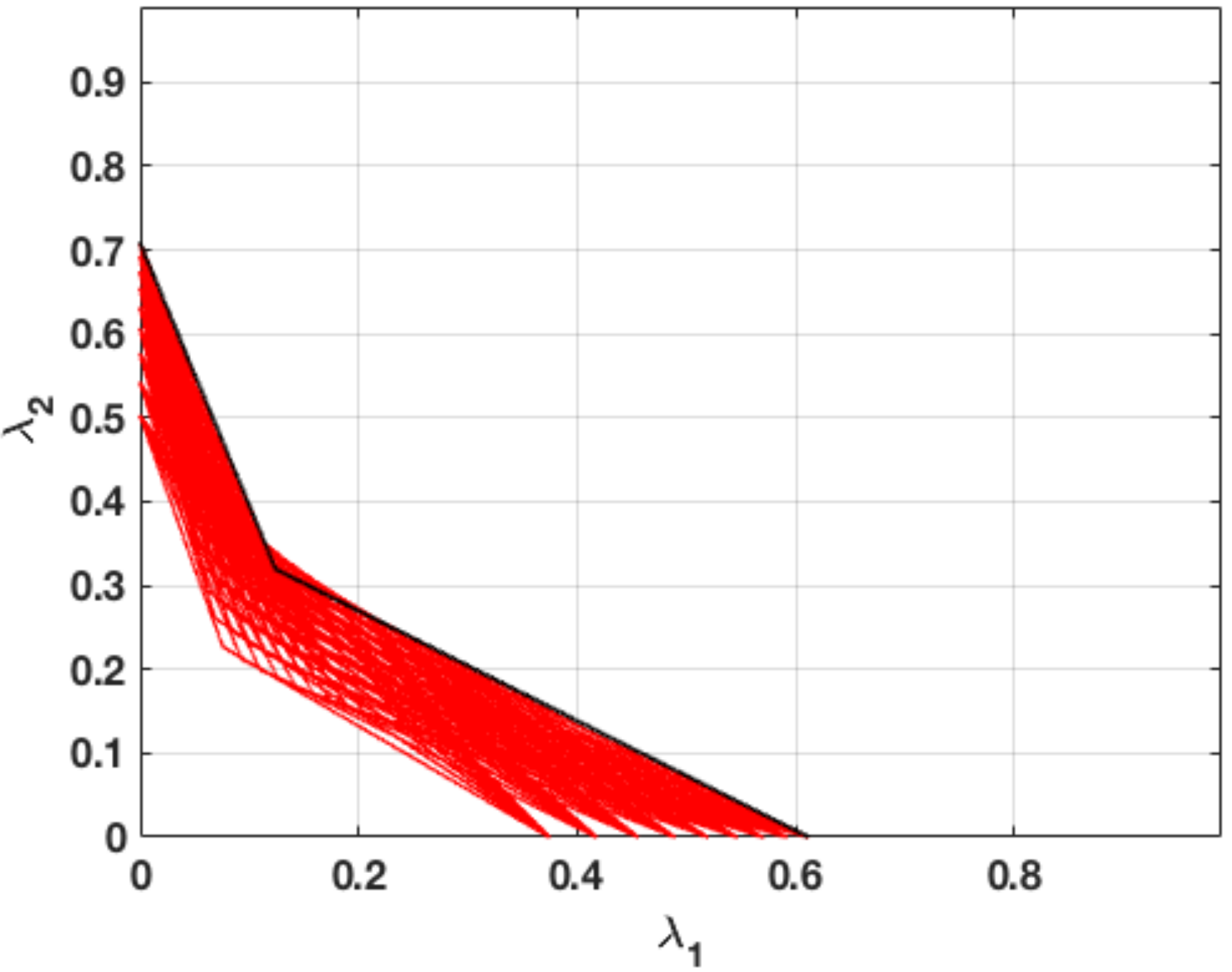}} \\
\subfigure[SIC,$\gamma_1=0.5,\gamma_2=0.4$]{\includegraphics[scale=0.5]{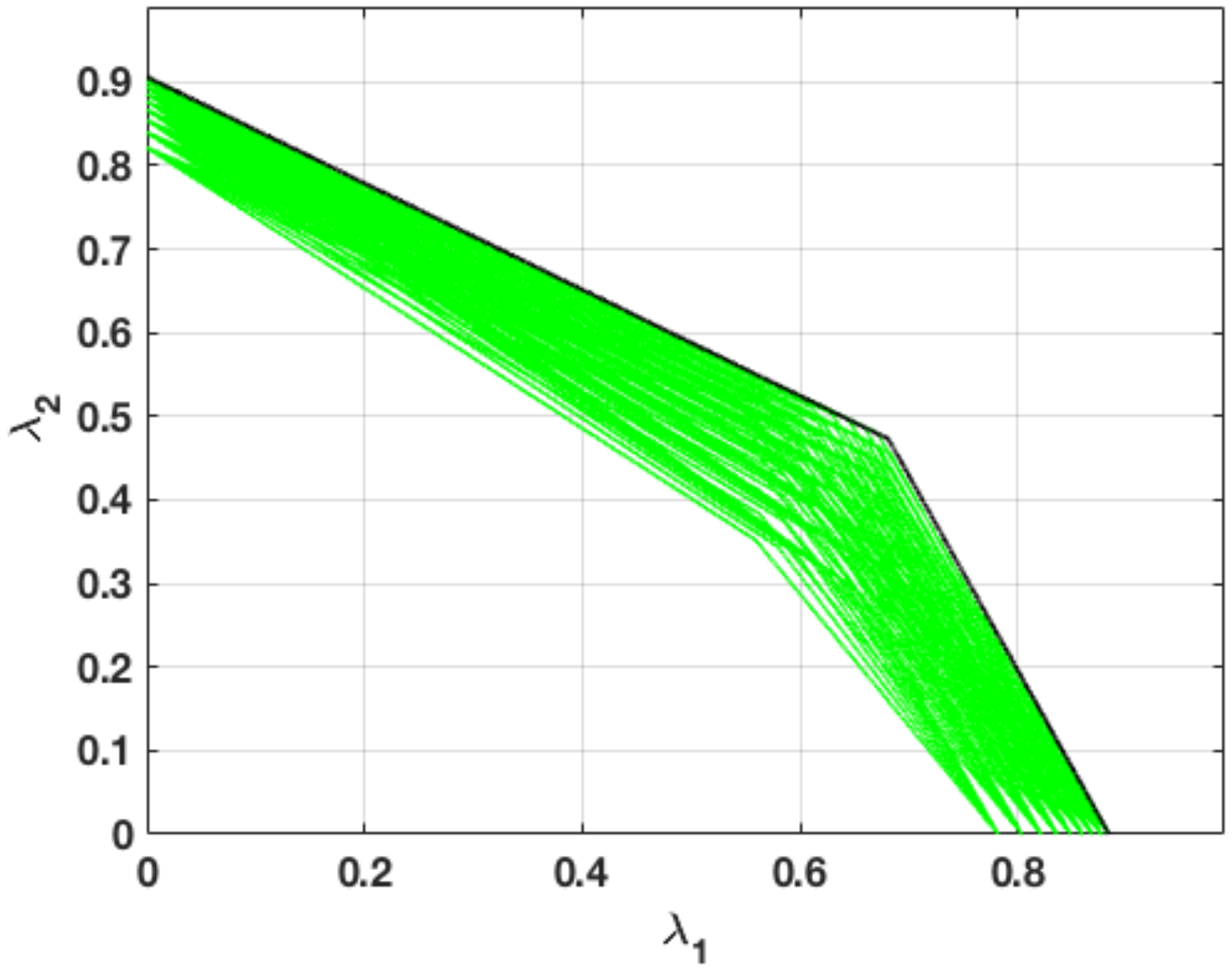}}
\subfigure[SIC,$\gamma_1=2,\gamma_2=1.4$]{\includegraphics[scale=0.5]{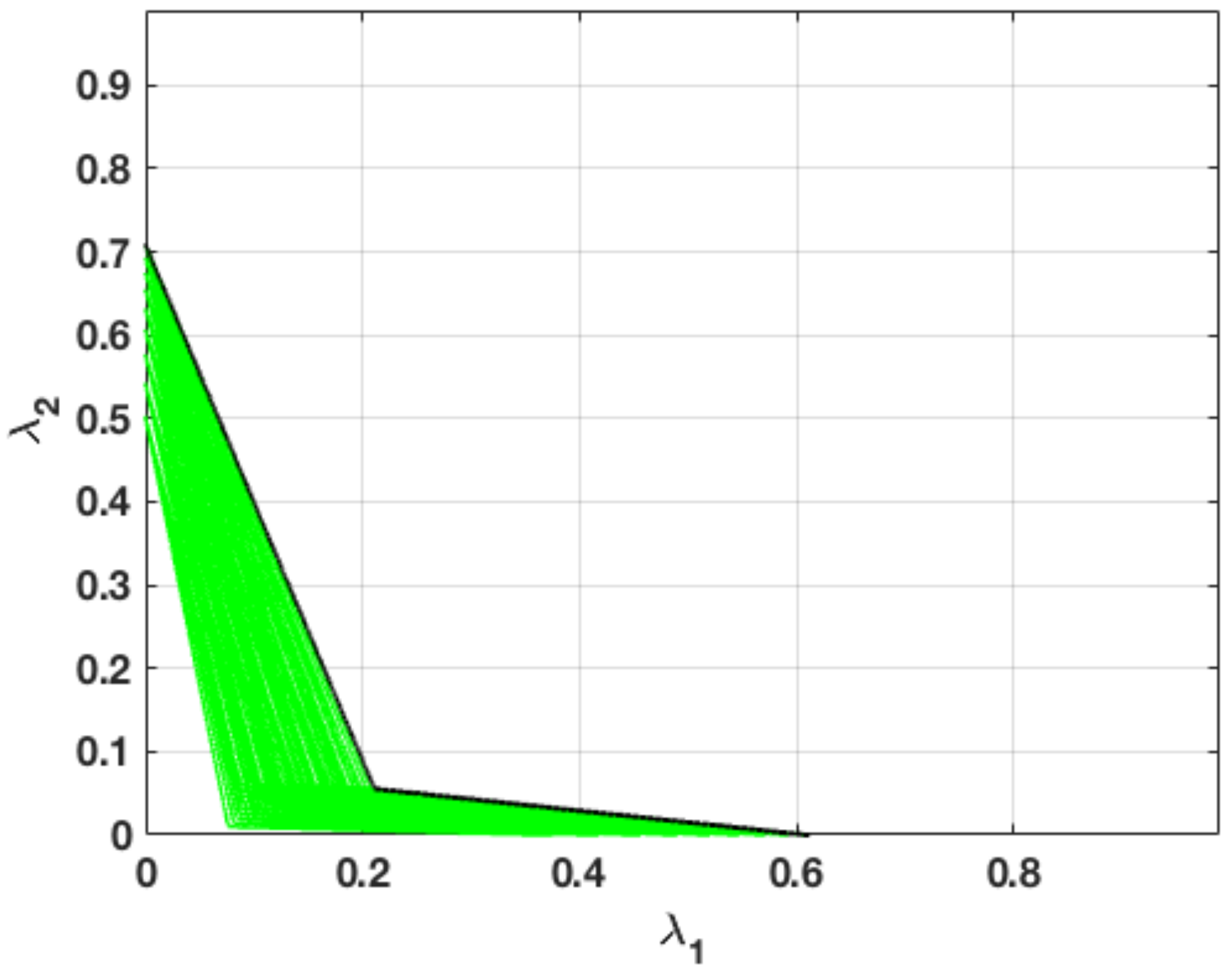}} \\
\subfigure[IAN-SIC,$\gamma_1=0.5,\gamma_2=0.4$]{\includegraphics[scale=0.5]{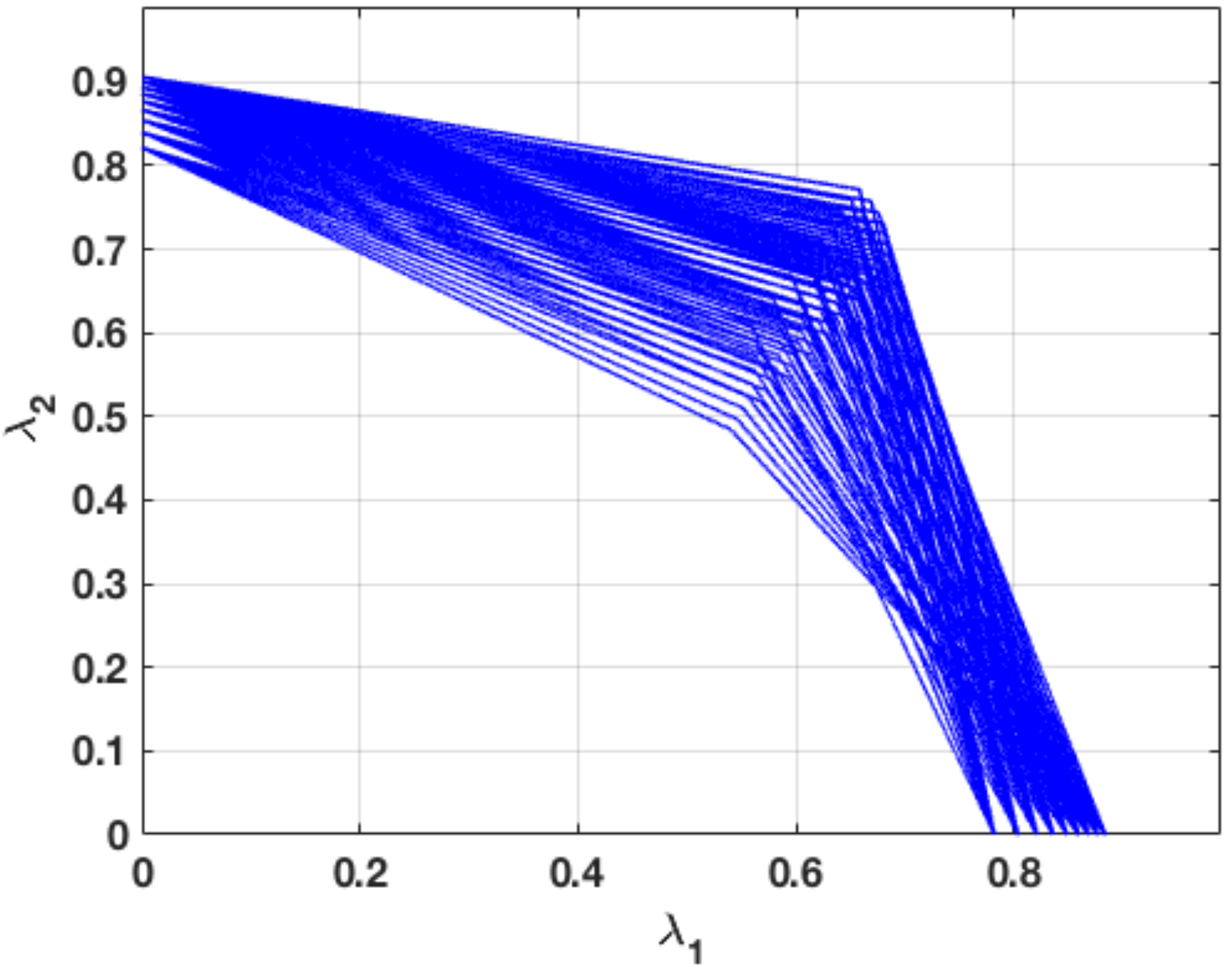}}	
\subfigure[IAN-SIC,$\gamma_1=2,\gamma_2=1.4$]{\includegraphics[scale=0.5]{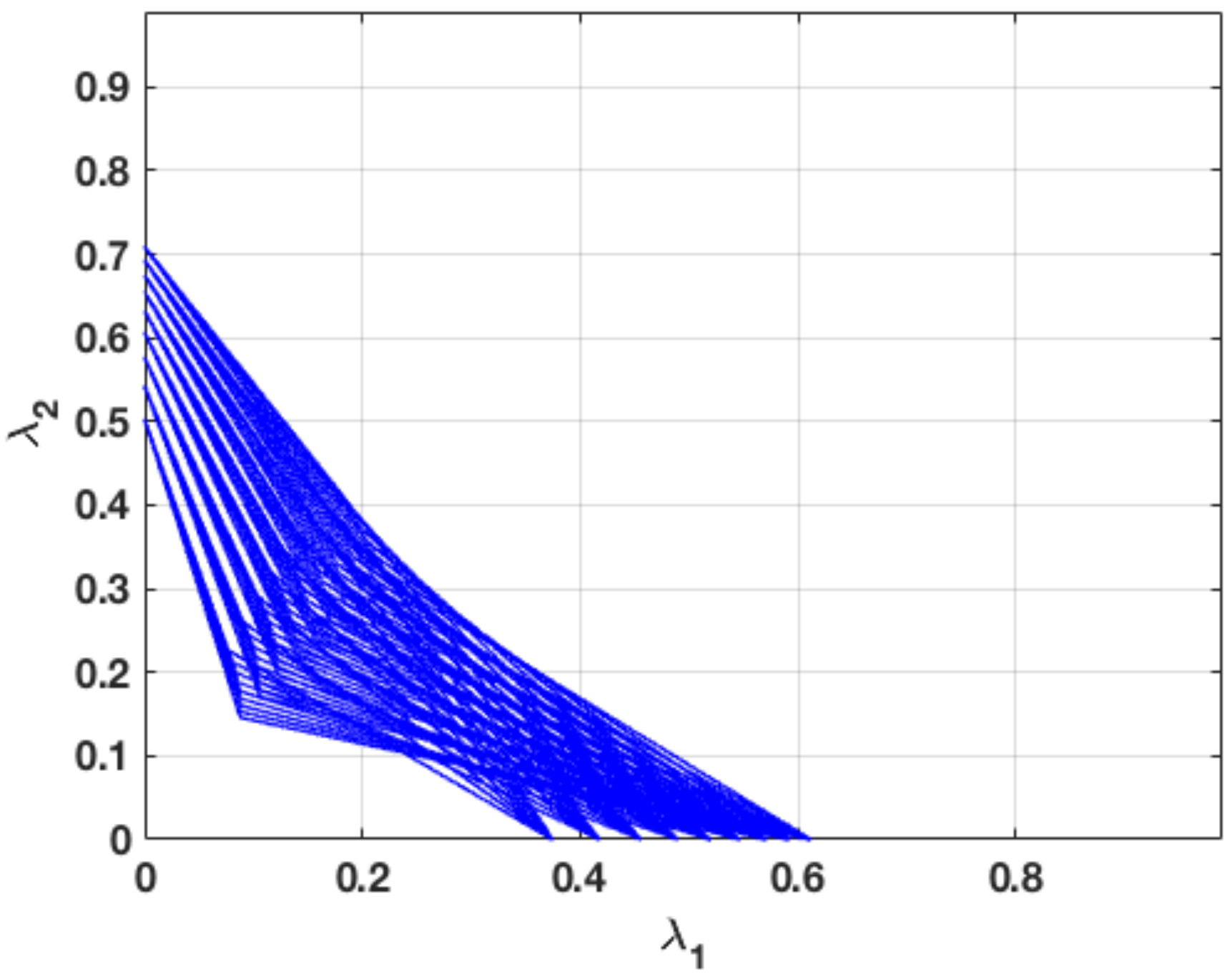}} \\
\caption{The closure of the stability region for the topology where $r_{11}=r_{22}=14, r_{12}=15, r_{21}=10$.}
\label{fig:Closure_Case2}
\end{figure*}

We observe that the black line defines the closure for the SIC for both cases of SINR thresholds. For the IAN, the black line approximates the closure; for low SINR thresholds there is a deviation in the area of high traffic regime for the first transmitter and the medium-to-low traffic regime for the second transmitter. For high SINR thresholds, the black line approximates well the closure for the IAN.

Another observation is that for higher SINR thresholds, all cases are non-convex sets, which is an indication that is better to apply time sharing. Nevertheless, the SIC-IAN scheme is closer to the performance of time sharing. For low SINR thresholds, there is no clear superior scheme, it depends on the traffic regimes.

\subsection{Closure of the stability region over all random access probabilities}
Here we obtain the closure of the stability region for the random access case over all transmission probabilities as defined in \eqref{eq:closure_def_RA_PA} for fixed transmit powers. More specifically, we consider the case where $p_1=600$ and $p_2=450$.

\subsubsection{First Topology: $r_{11}=r_{22}=10$, $r_{12}=r_{21}=5$}
The closure for the first topology is depicted in Fig. \ref{fig:Closure_RA_Case1}. We observe that the closure for the case where both receivers use SIC is a convex set. For $\gamma_1=0.5,\gamma_2=0.4$, we can approach the region that is achieved by transmitting with maximum power $p_1=p_2=p_{max}$, in Fig. \ref{fig:Closure_Case1}, by simply adjusting the transmission probabilities. This is an interesting observation that shows that random access can be an effective interference mitigation scheme compared to adjusting the transmission power. For the other two cases regarding the interference handling schemes we have similar observations as in Section \ref{sec:NumResClosurePower1}.

\begin{figure*}[!htbp]
\centering
\subfigure[IAN, $\gamma_1=0.5,\gamma_2=0.4$]{\includegraphics[scale=0.5]{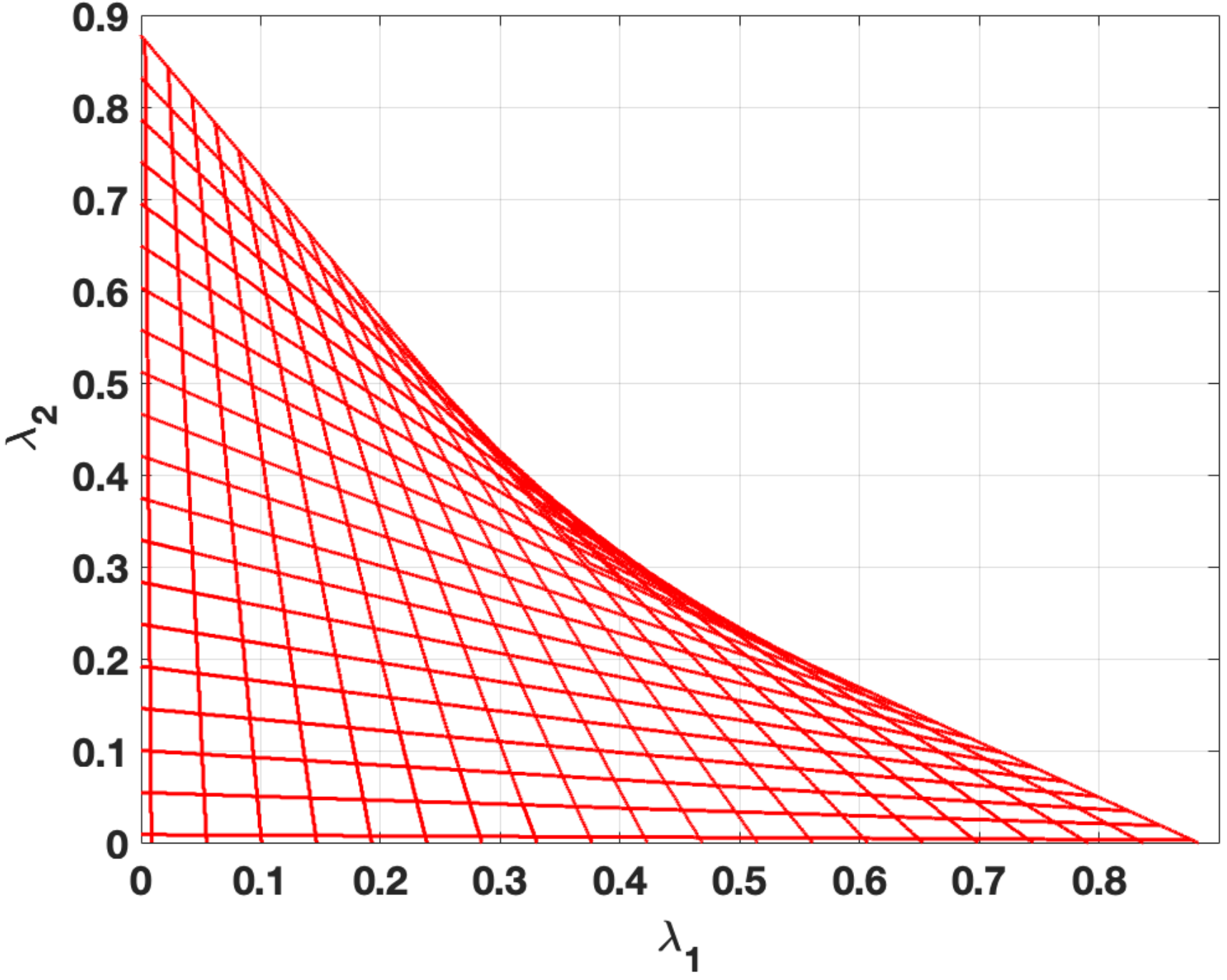}} 
\subfigure[IAN, $\gamma_1=2,\gamma_2=1.4$]{\includegraphics[scale=0.5]{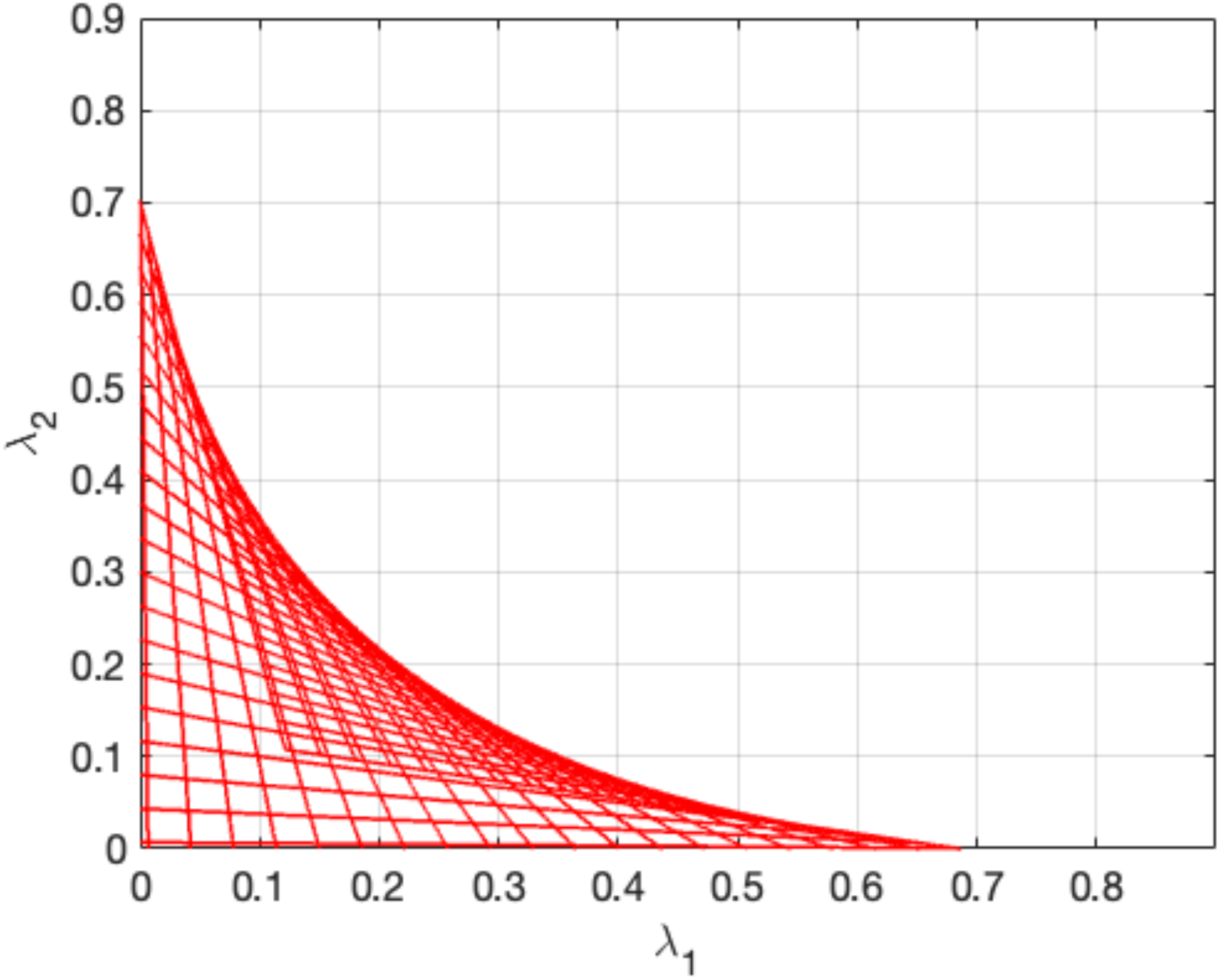}} \\
\subfigure[SIC,$\gamma_1=0.5,\gamma_2=0.4$]{\includegraphics[scale=0.5]{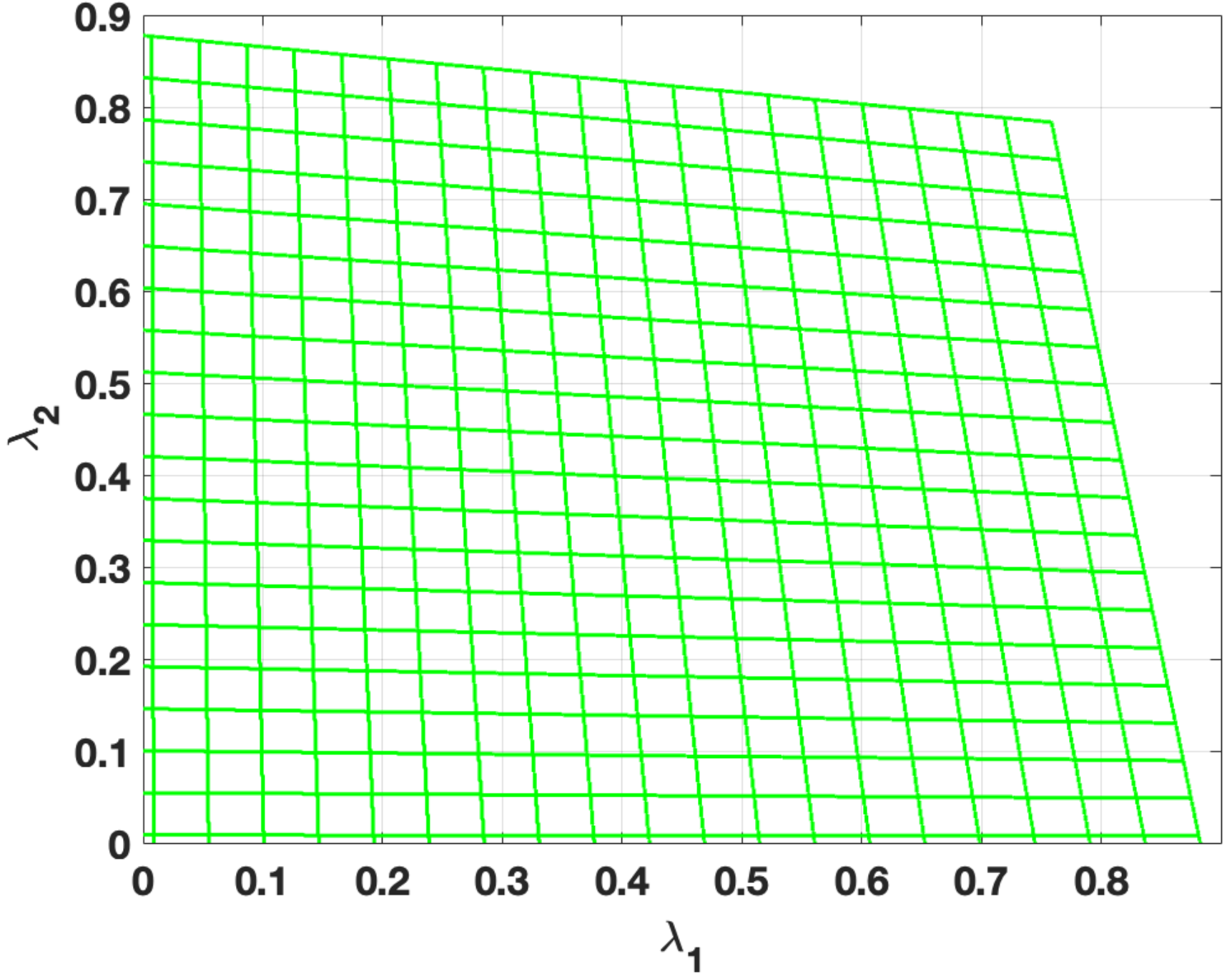}}
\subfigure[SIC,$\gamma_1=2,\gamma_2=1.4$]{\includegraphics[scale=0.5]{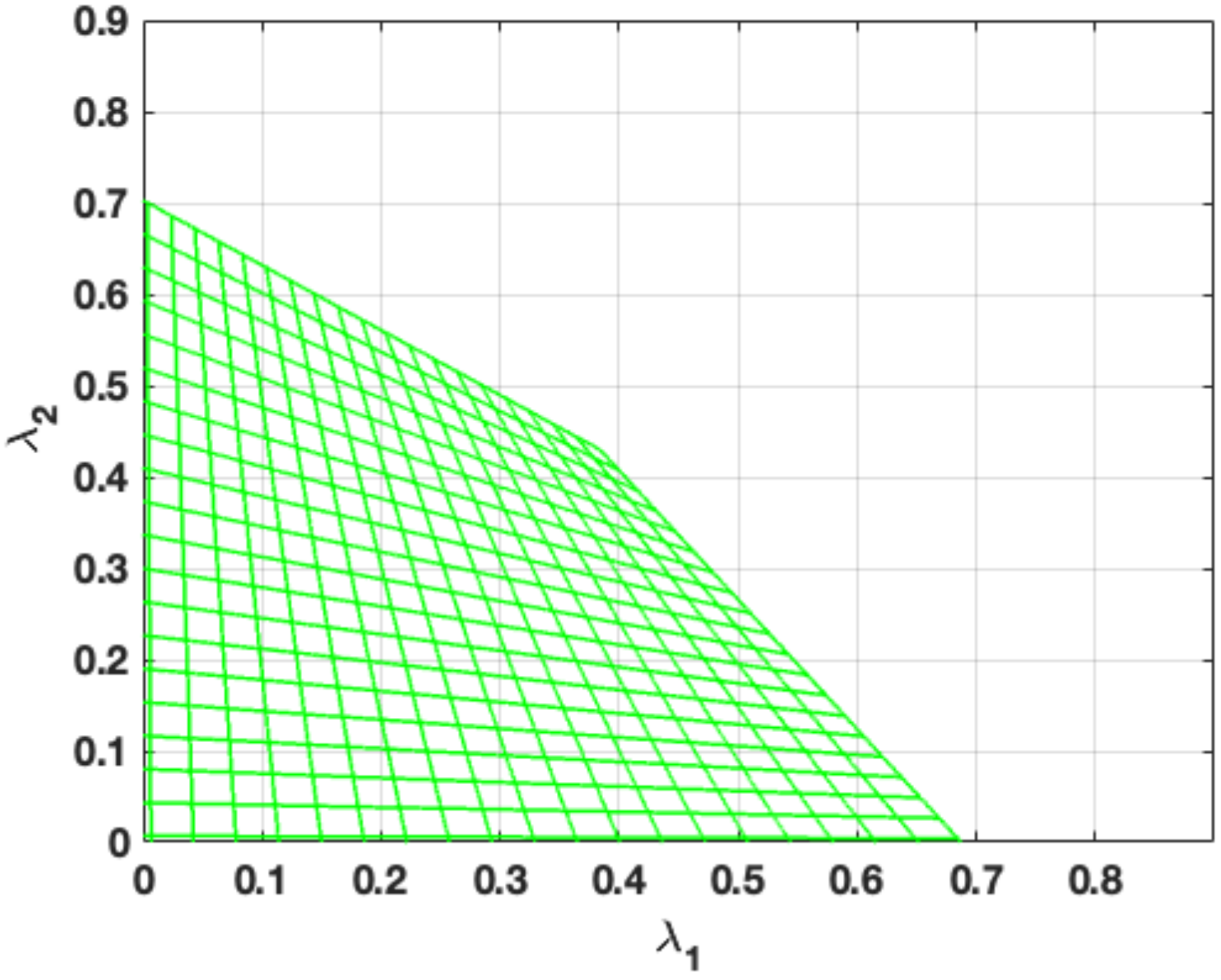}} \\
\subfigure[SIC-IAN,$\gamma_1=0.5,\gamma_2=0.4$]{\includegraphics[scale=0.5]{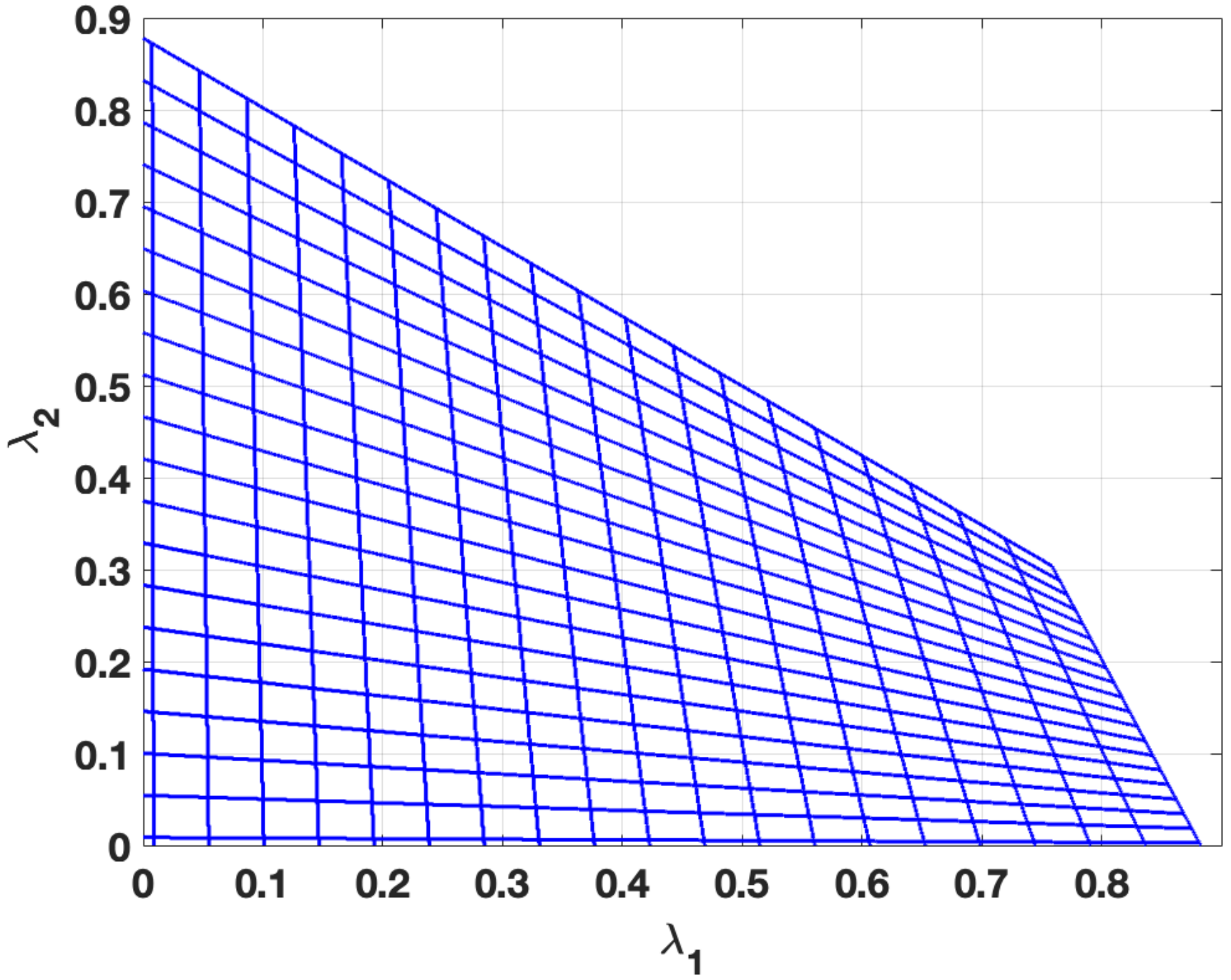}}	
\subfigure[SIC-IAN,$\gamma_1=2,\gamma_2=1.4$]{\includegraphics[scale=0.5]{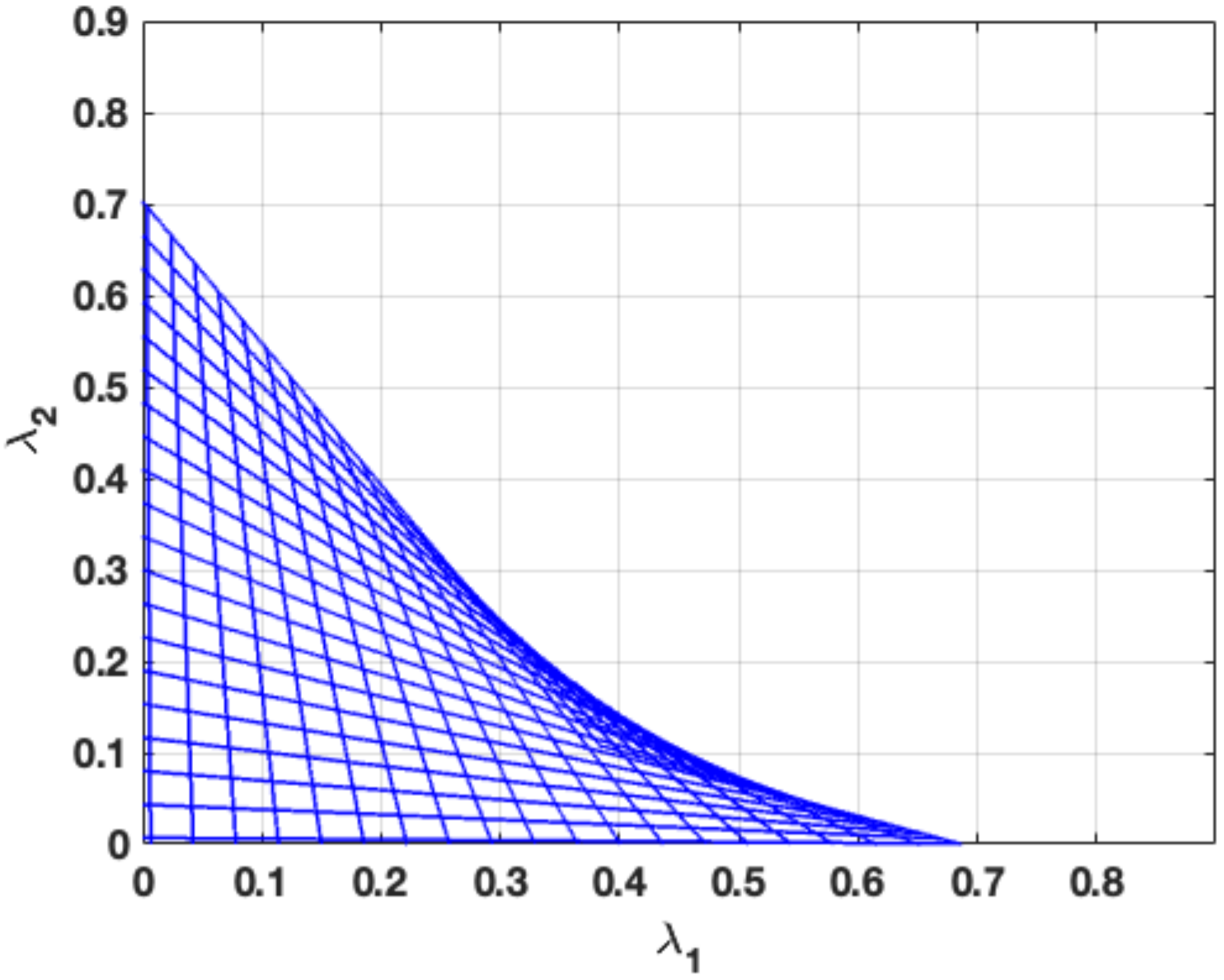}} \\
\caption{The closure over all access probabilities of the stability region for the topology where $r_{11}=r_{22}=10, r_{12}=r_{21}=5$.}
\label{fig:Closure_RA_Case1}
\end{figure*}

\subsubsection{Second Topology: $r_{11}=r_{22}=14$, $r_{12}=15, r_{21}=10$}

The closure for the second topology is depicted in Fig. \ref{fig:Closure_RA_Case2}. We observe that for the three schemes, for the low SINR thresholds, the closure is a convex set, which implies that they achieve better performance than time sharing. meaning that they achieve better performance than time sharing. We also observe that adjusting the access probabilities can be beneficial and it can have similar performance to increasing the transmit power. However, in the high values of the SINR thresholds, we observe that the closures are non-convex thus, time sharing schemes will are preferable.

\begin{figure*}[!htbp]
\centering
\subfigure[IAN, $\gamma_1=0.5,\gamma_2=0.4$]{\includegraphics[scale=0.5]{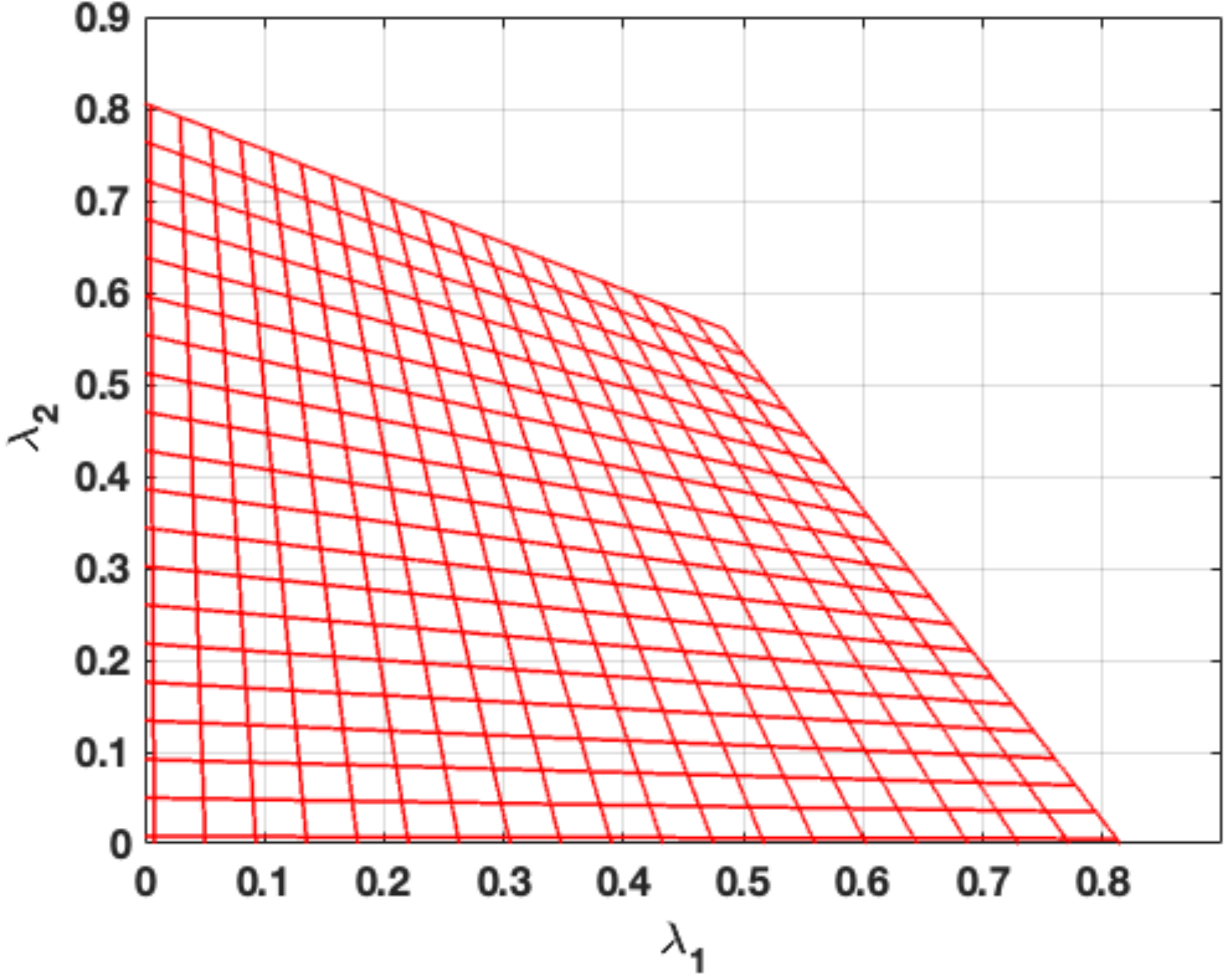}} 
\subfigure[IAN, $\gamma_1=2,\gamma_2=1.4$]{\includegraphics[scale=0.5]{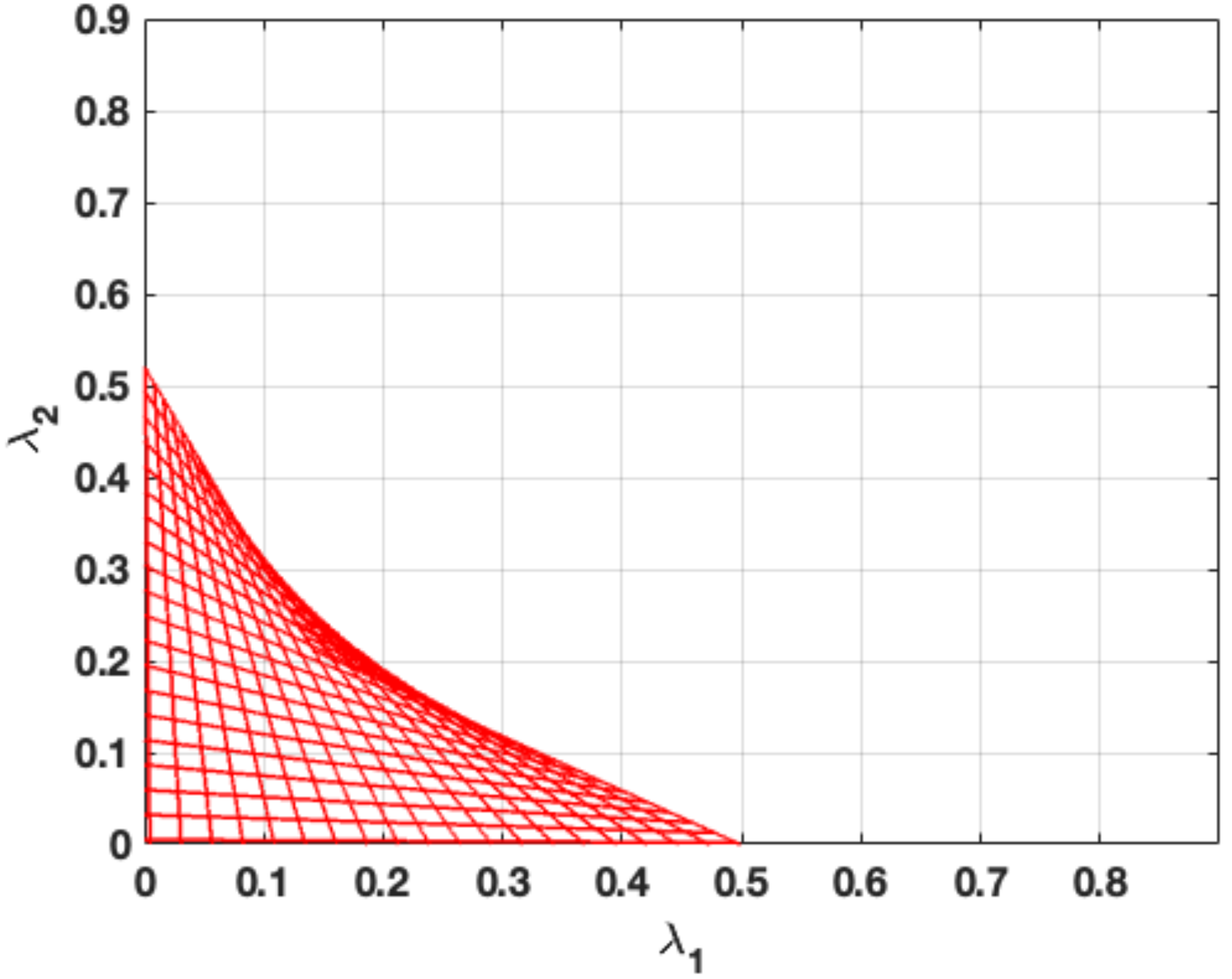}} \\
\subfigure[SIC,$\gamma_1=0.5,\gamma_2=0.4$]{\includegraphics[scale=0.5]{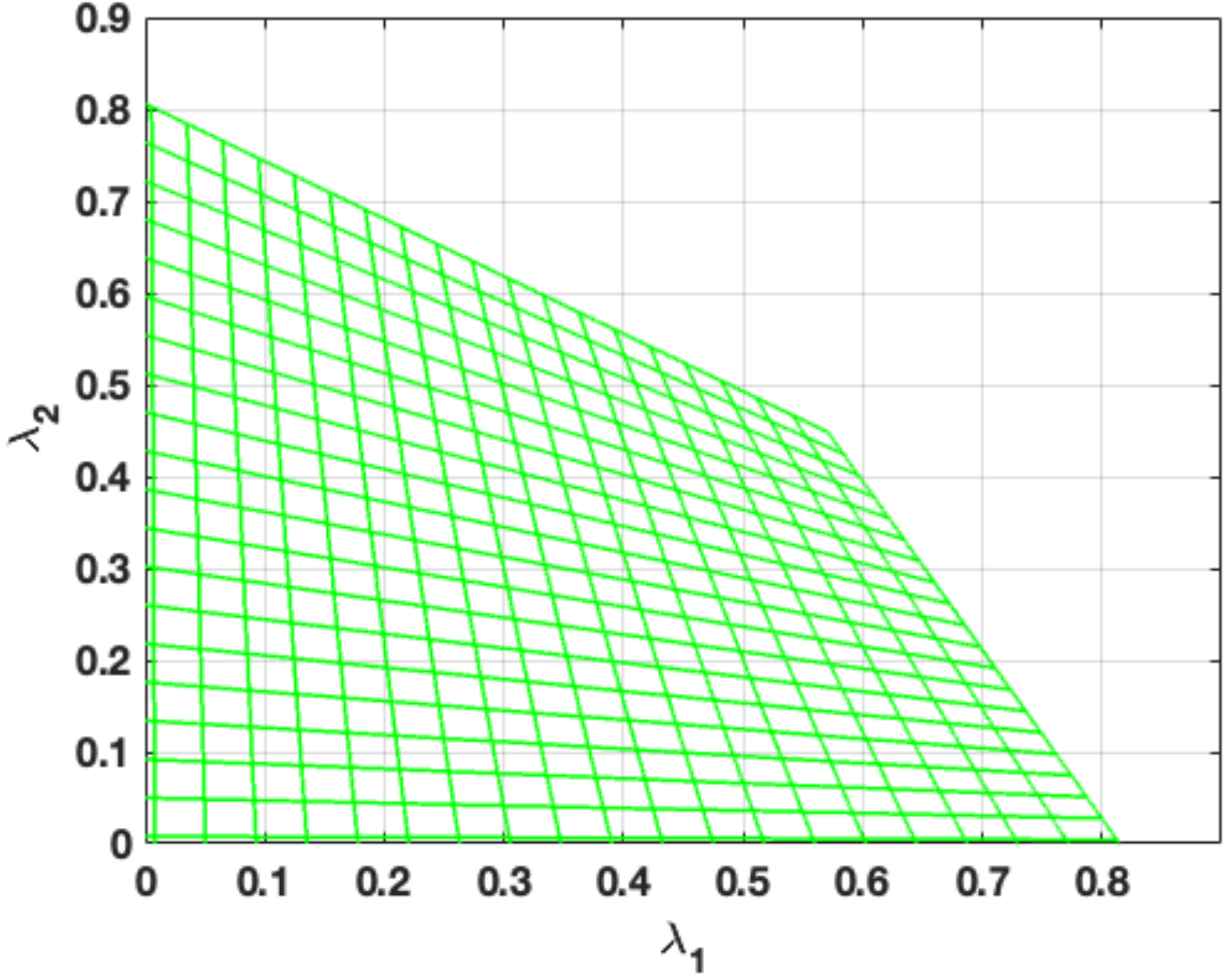}}
\subfigure[SIC,$\gamma_1=2,\gamma_2=1.4$]{\includegraphics[scale=0.5]{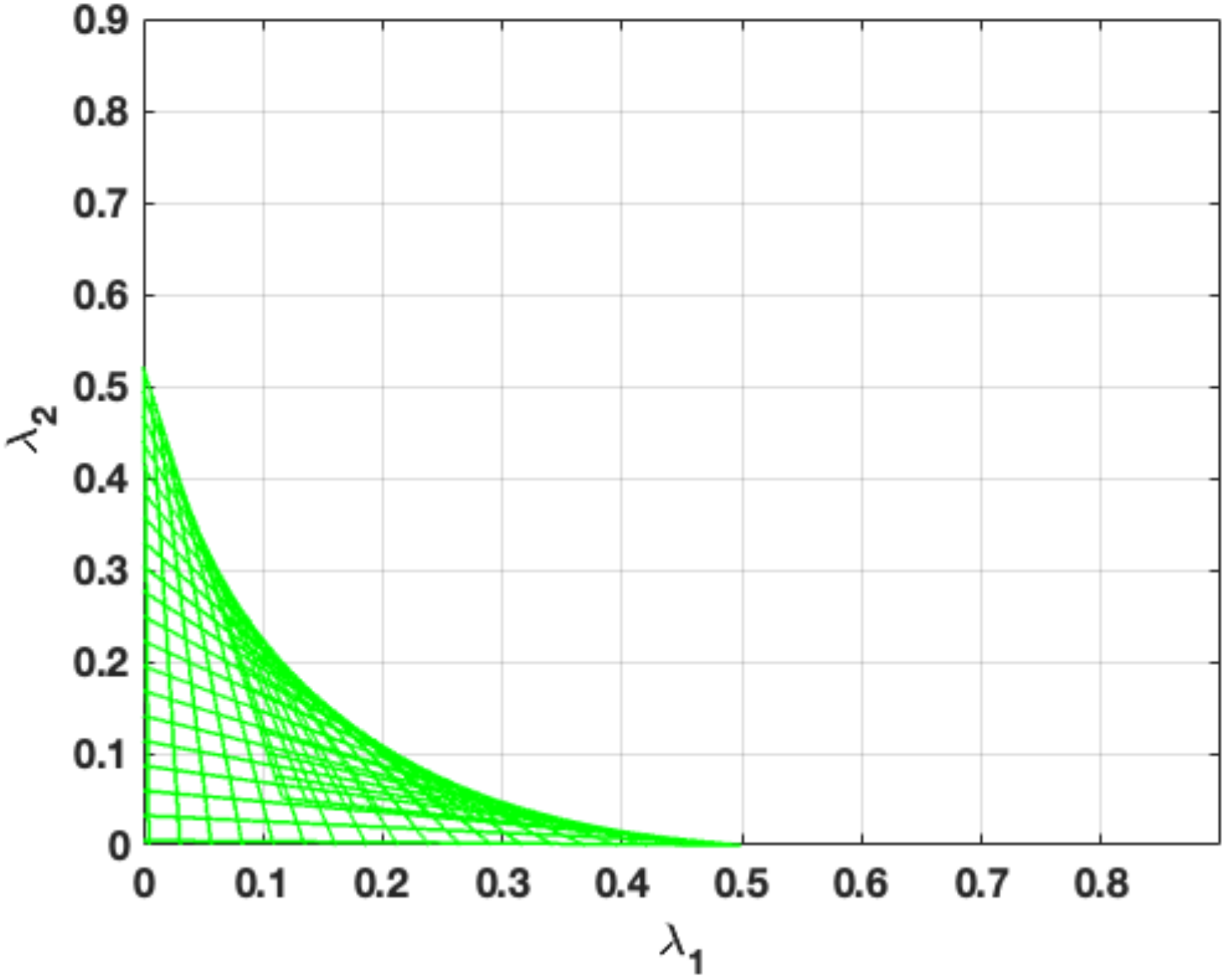}} \\
\subfigure[SIC-IAN,$\gamma_1=0.5,\gamma_2=0.4$]{\includegraphics[scale=0.5]{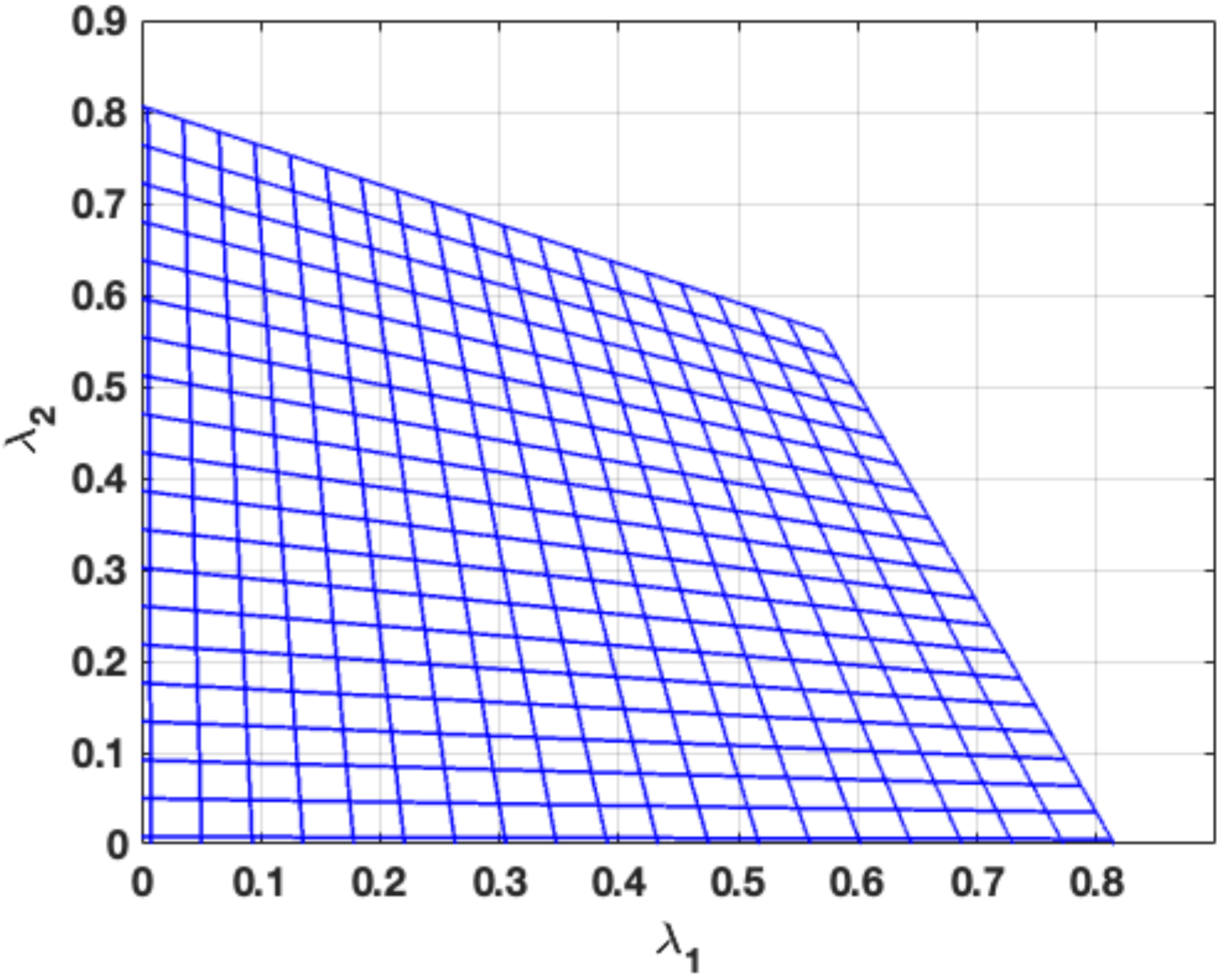}}	
\subfigure[SIC-IAN,$\gamma_1=2,\gamma_2=1.4$]{\includegraphics[scale=0.5]{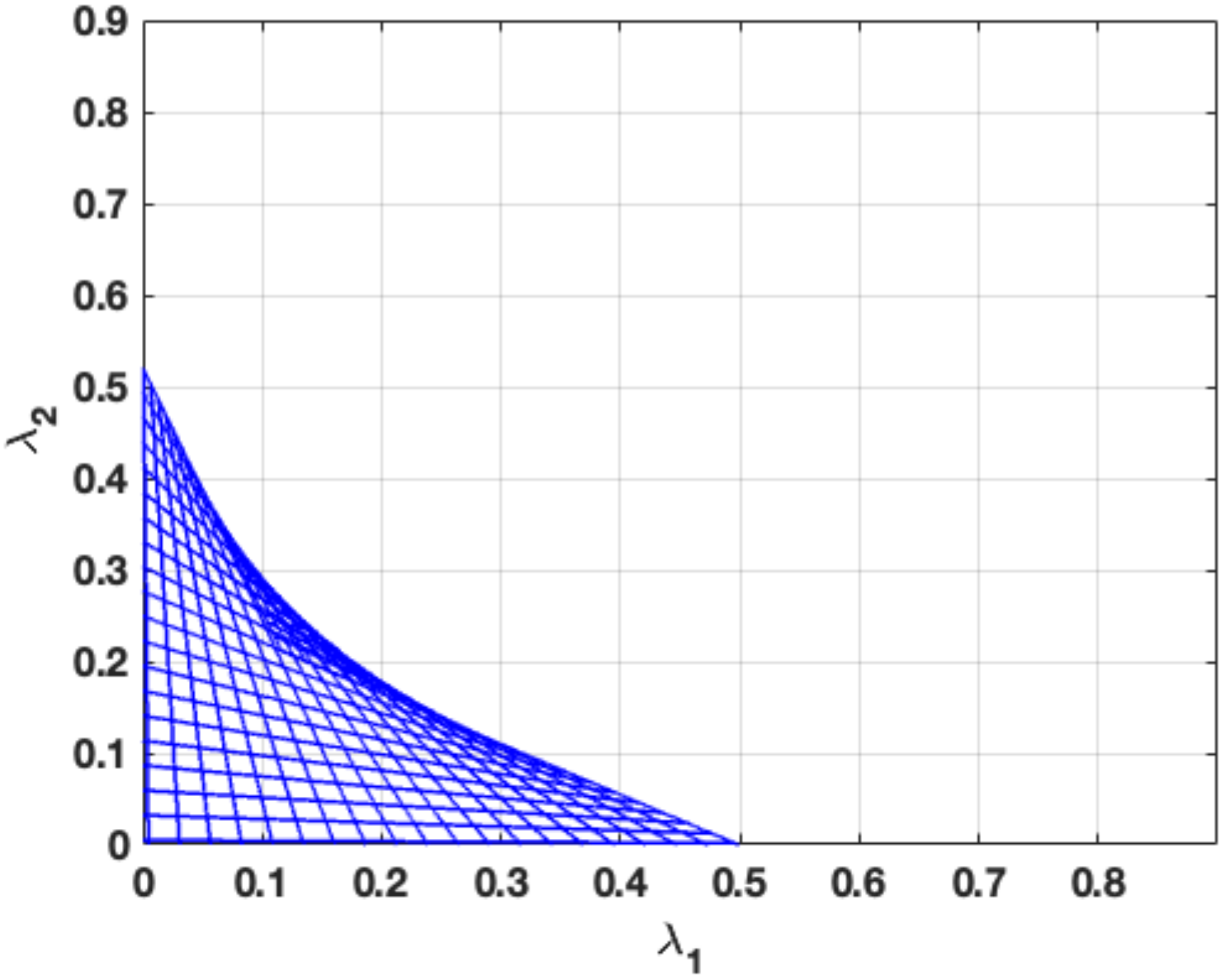}} \\
\caption{The closure over all access probabilities of the stability region for the topology where $r_{11}=r_{22}=14, r_{12}=15, r_{21}=10$.}
\label{fig:Closure_RA_Case2}
\end{figure*}

\section{Conclusions} \label{sec:conclusions}
In this work, we have derived the stability region of the two-user interference channel for both the general case and for two specific interference management strategies, namely treating interference as noise and successive interference cancellation at the receivers. Furthermore, we have provided conditions for the convexity of the stability region, as well as for ordering the interference management techniques depending on how broad the stability region achieved is. In addition, we have considered the effect of multiple transmit antennas and of random access on the stable throughput region.
 
Future extensions of this work may include the investigation of other random access schemes for the two-user interference channel, such as queue-aware schemes that can further increase the performance of the stable throughput combined with SIC. A future direction of this work is the investigation of the delay performance by utilizing the theory of boundary value problems. Another possible extension of this work is to consider the effect of successive interference cancellation in network level cooperative schemes.

\section*{Acknowledgements}
This work was supported in part by ELLIIT and the Center for Industrial Information Technology (CENIIT).

\bibliographystyle{IEEEtran}
\bibliography{thesis}

\end{document}